\documentclass[11pt, draftclsnofoot, onecolumn]{IEEEtran}
\usepackage[utf8]{inputenc}
\usepackage{tabularx}
\usepackage{amsmath}
\usepackage{amsthm}
\usepackage{amssymb}
\usepackage{enumerate}
\usepackage{latexsym}
\usepackage{mathtools}
\usepackage{amsfonts}
\usepackage{setspace}
\usepackage{multirow}
\usepackage[ruled,linesnumbered]{algorithm2e}

\usepackage{amsfonts,amsmath,amssymb,amsthm}
\usepackage{latexsym,amscd}
\usepackage{amsbsy}
\usepackage{CJK}
\usepackage{indentfirst,mathrsfs}
\usepackage{amsfonts,amsmath,amssymb,amsthm}
\usepackage{latexsym,amscd}
\usepackage{amsbsy}
\usepackage{color}
\usepackage{multirow}
\usepackage{makecell}
\usepackage{float}
\usepackage{cases}

\numberwithin{equation}{section}

\newtheorem{thm}{Theorem}[section]
\newtheorem{lem}{Lemma}[section]

\newtheorem{prop}{Proposition}[section]
\newtheorem{definition}{Definition}
\newtheorem{example}{Example}
\newtheorem{remark}{Remark}

\newcommand{\Tr}{{\rm {Tr}}}

\begin{document}

\title{A Novel Approach for Bent Functions with Dillon-like Exponents and Characterizing Three Classes of Bent Functions via Kloosterman Sums}
\author{Ziran Tu\thanks{Ziran Tu, Nian Li and Yanan Deng are with the Hubei Key Laboratory of Applied Mathematics, School of Cyber Science and Technology, Hubei
University, Wuhan, 430062, China. Email: tuziran@aliyun.com, nian.li@hubu.edu.cn, 1648481694@qq.com},
Sihem Mesnager\thanks{Sihem Mesnager is with the Department of Mathematics, University of Paris VIII, F-93526 Saint-Denis, Paris, France, the Laboratory Geometry, Analysis and Applications, LAGA, University Sorbonne Paris Nord CNRS,
UMR 7539, F-93430, Villetaneuse, France, and also with the Telecom Paris, 91120 Palaiseau, France. Email: smesnager@univ-paris8.fr},
Xiangyong Zeng\thanks{Xiangyong Zeng is with the Hubei Key Laboratory of Applied Mathematics, Faculty of Mathematics and Statistics, Hubei
University, Wuhan, 430062, China. Email: xzeng@hubu.edu.cn},
Nian Li, \\ Yupeng Jiang\thanks{Yupeng Jiang is with the School of Cyber Science and Technology, Beihang University, Beijing, 100191, China. Email:  jiangyupeng@amss.ac.cn},
Yanan Deng
}
\date{\today}
\maketitle

\begin{abstract}

Dillon-like Boolean functions are known, in the literature, to be those trace polynomial functions from $\mathbb{F}_{2^{2n}}$ to $\mathbb{F}_{2}$, with all the exponents being multiples of $2^n-1$ often called Dillon-like exponents. This paper is devoted to bent functions in which we study the bentness of some classes of Dillon-like Boolean functions connected with rational trace functions. Specifically,  we introduce a special infinite family of trace rational functions. We shall use these functions as building blocks and generalise notably a criterion due to Li et al. published in  [IEEE Trans. Inf. Theory 59(3), pp. 1818-1831, 2013]  on the bentness of Dillon-like functions in the binary case, we explicitly characterize three classes of bent functions. These characterizations are expressed in terms of the well-known binary Kloosterman sums. Furthermore, analysis and experiments indicate that new functions not EA-equivalent to all known classes of monomial functions are included in our classes.
\end{abstract}


\section{Introduction} \label{intro}

Bent functions are extraordinary Boolean functions with maximal distance from all affine functions. Rothaus \cite{rothaus} first defined and named these functions, and then Dillon gave a systematical study on them in his famous PhD thesis \cite{dillon1}. Of the interest in beautiful combinatorial objects and close relations with coding and cryptography, bent functions have attracted many people's attention. Up to now, very rich results on bent functions have been proposed, and a detailed survey is referred to as \cite{sihembook}. The research on bent functions has lasted about five decades. It is still active because of our interest in these functions {\color{red}from both} theoretical and practical aspects in several different (but usually related) areas. Different generalisations, including vectorial bent functions, plateaued functions \cite{carlet3,zhengyuliang}, semi-bent functions \cite{charpin3}, hyper-bent  \cite{youssef}, and nega-bent functions \cite{schmidt}, $p$-ary bent functions \cite{kumar}, etc., have been proposed and much progress has been made. Since it still seems elusive to classify bent functions completely, the constructions, characterizations and generations constitute the main relative research activities. An interesting topic is constructing bent functions in polynomial forms \cite{canteaut,charpin1,charpin3,dillon1,dobbertin,yu}. As far as we know, there are only five classes of monomial bent functions (up to Extended-Affine (EA) equivalence), i.e., bent functions $\Tr_{1}^{2n}(ax^d)$ with suitable $a\in\mathbb{F}_{2^{2n}}^{*}$. We call these integers $d$ bent exponents (see Table \ref{tab2}).

\begin{table}[!htb]
\footnotesize
\caption{{The known bent exponents on $\mathbb{F}_{2^{2n}}$}\label{tab2}}
\begin{center}
\begin{tabular}{llcc}\hline  \hline
Case & Exponents $d$ & Algebraic degree  & References\\\hline
Gold & $2^i+1$ & 2 & \cite{gold} \\
Dillon & $l(2^n-1)$  &$n$&  \cite{charpin2,dillon2,lachaud,leander} \\
Leander & $(2^{\frac{n}{2}}+1)^2$  &3&  \cite{charpin1,leander} \\
Kasami & $2^{2i}-2^i+1$  &$i+1$&  \cite{dillon1} \\
Canteaut-Charpin-Kyureghyan & $2^{\frac{2n}{3}}+2^{\frac{n}{3}}+1$   &3& \cite{canteaut} \\
\hline
\end{tabular}
\end{center}
\end{table}
Dillon was the first to characterize the bentness of the function $\Tr_{1}^{2n}(ax^{2^n-1})$, then Charpin and Gong \cite{charpin2} and Leander \cite{leander} generalized the function to $\Tr_{1}^{2n}(ax^{l(2^n-1)})$ with $(l,2^n+1)=1$. Dillon-like functions with multiple trace terms were further investigated in several later works. Using Dickson polynomials and some exponential sums, the authors in \cite{charpin2} gave a characterization on the bentness of Boolean functions of the form $\sum_{i\in E}\Tr_{1}^{2n}(a_{i}x^{i(2^n-1)})$ under some restrictions. Mesnager extended the Charpin-Gong-like functions with an additional trace term over $\mathbb{F}_{4}$, i.e., $f_{a,b}^{r}=\Tr_{1}^{2n}(ax^{r(2^n-1)})+\Tr_{1}^{2}(bx^{\frac{2^{2n}-1}{3}})$ \cite{sihem2,sihem3}. Li et al. in \cite{nianli} further considered functions with the following type
\begin{equation}\label{form1}
f(x)=\sum_{i=0}^{2^n-1}\Tr_{1}^{2n}(a_{i}x^{i(2^n-1)})+\Tr_{1}^{o(e)}(\epsilon x^{\frac{2^{2n}-1}{e}}),
\end{equation}
where $e\,|\,2^n+1$, $a_i\in \mathbb{F}_{2^{2n}}$, $o(e)$ is the smallest positive integer satisfying $o(e)\,|\,2n$ and $e\,|\,2^{o(e)}-1$. They gave a necessary and sufficient condition to characterize the bentness. However, it is desired to give more explicit characterizations. Some restrictions on these coefficients $a_{i}$ and $\epsilon$, or on the integers $i$ and $e$ might be necessary to achieve this. For example, in \cite{nianli}, several new binomial, trinomial and quadrinomial bent functions were found by restricting the number of terms. The authors also considered functions with multiple trace terms satisfying that all coefficients $a_{i}$ were equal. In \cite{yulong}, functions having the form $$\sum_{i=0}^{d-1}\Tr_{1}^{m}(a_{i}x^{(l+i\frac{2^n+1}{d})(2^n-1)})+\Tr_{1}^{o(d)}(\epsilon x^{\frac{2^{2n}-1}{d}})$$ had been considered, where the coefficients $a_{i}\in \mathbb{F}_{2^n}$, $a_{1}=a_{2}=\ldots=a_{d-1}$ and $a_{0}\ne a_{1}$.

We observe that functions defined in \eqref{form1} over $\mathbb{F}_{2^{2n}}$ actually have a more generalized form as
\begin{equation}\label{form2}
f(x)=g(x^{2^n-1}),
\end{equation}
where $g$ is a function from $\mathbb{F}_{2^{2n}}$ to $\mathbb{F}_{2^{2n}}$ satisfying $g(y)\in \mathbb{F}_{2}$ for $y\in\{0\}\cup\{y\in \mathbb{F}_{2^{2n}}^{*}:y^{2^n+1}=1\}$. In \cite{youssef}, Charpin and Gong proposed an important subclass of bent functions called hyper-bent. Then Carlet and Gaborit described the result of Charpin and Gong from the view of Boolean functions \cite{carlet2}. We note that the form \eqref{form2} coincides with a requirement in that description.

In this paper, we introduce a special family of functions
$$\mathcal{H}=\left\{\Tr_{1}^{2n}\left(\frac{a}{x^{2^n-1}+b}\right):\,a,b\in \mathbb{F}_{2^{2n}}^{*}\right\}.$$ These functions in set $\mathcal{H}$ are trace rational. They are also Dillon-like if seen as ordinary trace polynomials. Then by using such kind of functions as building blocks and generalizing the criteria in \cite{nianli} on the binary case, we consider the bentness of functions with the form
\begin{equation}\label{form3}
h(x)=F(f_{1}(x),f_{2}(x),\ldots,f_{t}(x))
\end{equation}
where $F(X_{1},X_{2},\ldots,X_{t})$ is a reduced polynomial in $\mathbb{F}_{2}[X_{1},X_{2},\ldots,X_{t}]$. It is interesting that, by choosing suitable block functions $f_{i}$, $i=1,2,3$ and $F(X_{1},X_{2},X_{3})=X_{1}X_{2}+X_{1}X_{3}+X_{2}X_{3}$, Carlet had considered this method to generate $k$-resilient functions \cite{carlet}. When letting $f_{i}(x)=\Tr_{1}^{2n}(u_{i}x)$, Tang et al. used such functions $h$ to modify certain functions to obtain new bent functions \cite{tangchunming}. Our main contribution here are characterizations of three classes of such bent functions with $F(X_{1},X_{2},\ldots,X_{t})\in \{X_{1}, X_{1}X_{2}, X_{1}X_{2}+X_{1}X_{3}+X_{2}X_{3}\}$. As far as we know, this is the first time bent functions connected with this form of rational trace polynomials have been investigated. To determine the bentness of $h_{1}(x)=F(f_{1})$ with $F(X_{1})=X_{1}$, i.e., $h_{1}=\Tr_{1}^{2n}(\frac{a_{1}}{x^{2^n-1}+b})$ with $a_{1}\in \mathbb{F}_{2^n}^{*}$ and $b\in\mathbb{F}_{2^{2n}}^{*}$, the key is computing exponential sums concerning some rational functions, of which the numerator has degree one and the denominator has degree two. For the other two classes of functions with $F\in \{X_{1}X_{2}, X_{1}X_{2}+X_{1}X_{3}+X_{2}X_{3}\}$, the technique is decomposing certain exponential sums into a combination of the sums connected with $h_{1}$. We finally give some discussion on the EA equivalence, and find that in our classes there are new bent functions which are not EA-equivalent to all known classes of monomial bent functions.

The rest of this article is organized as follows. Section \ref{sec2} gives some notations and preliminaries, including the criterion on bentness of functions in \eqref{form2} and some useful exponential sums. In Section \ref{mainresults}, we describe our method of constructing bent functions and characterising three classes of bent functions. Section \ref{sec4} discusses the EA equivalence, and Section \ref{sec4}  concludes the paper.

\section{Preliminaries}\label{sec2}
\subsection{Notations and background}
An $m$-variable Boolean function $f$ is a mapping from the vector space $\mathbb{F}_{2}^m$ to the binary field $\mathbb{F}_{2}$. Since the vector space $\mathbb{F}_{2}^m$ is usually identified with the finite field $\mathbb{F}_{2^m}$, we also treat $f$ as a mapping from $\mathbb{F}_{2^m}$ to $\mathbb{F}_{2}$. The set of all $m$-variable Boolean functions is denoted by $\mathcal{B}_{m}$. The \textit{support} of $f$ is the set $supp(f)=\{x\in \mathbb{F}_{2}^m|f(x)=1\}$. The cardinality of $supp(f)$ is called the \textit{Hamming weight} of $f$, which is denoted by $wt(f)$.
There are two general ways to represent a Boolean function. The first one is called the \textit{Algebraic Normal Form} (ANF), i.e., we represent $f$ as an $m$-variable polynomial
$$f(x_{1},x_{2},\ldots,x_{m})=\bigoplus_{I\in \mathcal{P(M)}}a_{I}\prod_{i\in I} x_{i},$$
where $a_{I}\in \mathbb{F}_{2}$ and $\mathcal{P(M)}$ denotes the power sets of $\mathcal{M}=\{1,2,\ldots,m\}$. Such representation belongs to the quotient ring $\mathbb{F}_{2}[x_{1},x_{2},\ldots,x_{m}]/(x_{1}^2\oplus x_{1},x_{2}^2\oplus x_{2},\ldots,x_{m}^2\oplus x_{m})$ and is unique. A polynomial $F(X_{1},X_{2},\ldots,X_{t})$ in $\mathbb{F}_{2}[X_{1},X_{2},\ldots,X_{t}]$ is called \textit{reduced} if it is the ANF of some Boolean function $f$. The \textit{algebraic degree} of $f$ is defined as $deg(f)=max\{\,|I|:a_{I}\ne 0\}$.
For two positive integers $m$ and $n$ with $n\,|\,m$, we use ${\rm
Tr}_{n}^{m}(\cdot)$ to denote the {\it trace function} from
$\mathbb{F}_{2^m}$ to $\mathbb{F}_{2^n}$ \cite{lidl}, i.e.,
$${\rm Tr}_{n}^{m}(x)=x+x^{2^n}+x^{2^{2n}}+\cdot\cdot\cdot+x^{2^{(m/n-1)n}}.$$
Then each non-zero Boolean function $f$ from $\mathbb{F}_{2^m}$ to $\mathbb{F}_{2}$ has a unique polynomial representation \cite{sihembook}, i.e.,
$$f(x)=\sum_{j\in \Gamma_{m}}\Tr_{1}^{\mathfrak{o}(j)}(a_{j}x^j)+\epsilon(1+ x^{2^m-1}),~a_{j}\in \mathbb{F}_{2^{\mathfrak{o}(j)}},~$$
where
\noindent\begin{itemize}
\item[-] $\Gamma_{m}$ denotes the set of integers by choosing exactly one element in each cyclotomic coset modulo $2^m-1$, including the trivial coset which only contains $0$;
\item[-] $\mathfrak{o}(j)$ is the size of the cost containing $j$;
\item[-] $\epsilon = wt(f) ~{\rm mod }~2$.
\end{itemize}
 %Given a basis $\alpha_{1},\alpha_{2},\ldots,\alpha_{m}$ of $\mathbb{F}_{2^m}$ over $\mathbb{F}_{2}$, the ANF and the polynomial form of $f$ can be transformed to each other.
The Walsh transform of $f$ at $\omega\in \mathbb{F}_{2^m}$ is defined as
$$W_{f}(\omega)=\sum_{x\in \mathbb{F}_{2^m}}(-1)^{f(x)+\Tr_{1}^{m}(\omega x)}.$$
The following is the definition of bent functions.
\begin{definition}
A Boolean function $f\in \mathcal{B}_{m}$ is called bent if $W_{f}(\omega)^2=2^m$ for all $\omega\in \mathbb{F}_{2^m}$.
\end{definition}
 Bent functions exist only for even $m$. In the following, we always assume $m=2n$ with $n$ being a positive integer. A notable fact is that bent functions always occur in pairs. If a Boolean function $f$ is bent, another bent function $\widetilde{f}$ is defined by considering the sign of the values $\{W_{f}(\alpha):\alpha\in \mathbb{F}_{2^m}\}$, i.e., $(-1)^{\widetilde{f}(x)}=2^{-n}W_{f}(x)$. We call $\widetilde{f}$ the \textit{dual} of $f$. In many cases, given a bent function $f$, the computation of its dual is not an easy task.
An important subset of $\mathbb{F}_{2^{m}}$ is the unit circle. For each element $x$ in the finite field $\mathbb{F}_{2^{m}}$, denote
$\overline{x}=x^{2^n}$.
The {\it unit circle} of $\mathbb{F}_{2^{m}}$ is defined as
\begin{equation}\label{defU}
U=\left\{\,\eta\in \mathbb{F}_{2^{m}}: \eta^{2^n+1}=\eta\overline{\eta}=1\right\}.
\end{equation}
The so-called polar decomposition, i.e., each $x\in \mathbb{F}_{2^{m}}^{*}$ can be uniquely written as $x=\lambda y$ with $\lambda\in U$ and $y\in \mathbb{F}_{2^n}^{*}$, is very useful. There is a correspondence between $U$ and the subfield $\mathbb{F}_{2^n}$, which can be used to transform the exponential sum on $U$ to $\mathbb{F}_{2^n}$.
\begin{lem}\label{lem0} {\rm\cite{Rosendahl}} Let $m=2n$ be a positive integer, $A\in \mathbb{F}_{2^{m}}\setminus\mathbb{F}_{2^n}$ be fixed and $U$ be given in Eq.\eqref{defU}. Then
$$U\setminus \{1\}=\left\{\frac{u+{A}}{u+\overline{A}}: u\in \mathbb{F}_{2^n}\right\}.$$
\end{lem}

The next lemma deals with the number of solutions in $U$ of a specific quadratic equation over a finite field of even characteristic \cite{alah,dodunekov,ztu}.
\begin{lem}\label{lem2.1}
Let $m=2n$ be an even positive integer and $a,b\in \mathbb{F}_{2^m}^{*}$ satisfying ${\rm Tr}_{1}^{m}\left(\frac{b}{a^2}\right)=0$. Then the
quadratic equation $x^2+ax+b=0$ has

{\rm (1)} both solutions in the unit circle if and only if
$$\begin{array}{c}
b=\frac{a}{\bar{a}} \text{ and }
{\rm Tr}_{1}^{n}\left(\frac{b}{a^2}\right)={\rm Tr}_{1}^{n}\left(\frac{1}{a\bar{a}}\right)=1;
\end{array}
$$

{\rm (2)} exactly one solution in the unit circle, if and only if
$$
\begin{array}{c}
b\ne \frac{a}{\bar{a}} \text{ and }
(1+b\bar{b})(1+a\bar{a}+b\bar{b})+a^2\bar{b}+\bar{a}^2b=0.
\end{array}
$$
\end{lem}

\subsection{On Dillon-like functions of the form  $f=g(x^{2^n-1})$}
One of the most important classes of bent functions is the \textit{Partial Spread Class}($\mathcal{PS}$) \cite{dillon1}, which are defined relative to classic geometric objects called partial spreads.
\begin{definition}
Let $m=2n$ be an even integer. A partial spread of $\,\mathbb{F}_{2^m}$ is a set
of pairwise supplementary $n$-dimensional subspaces of $\mathbb{F}_{2^m}$. A partial spread is a
spread if the union of its elements equals $\mathbb{F}_{2^m}$.
\end{definition}

Boolean functions $f(x)=g(x^{2^n-1})$ constitute a special subclass of the famous partial spread functions, and the corresponding spread is the Desarguesian spread \cite[page 346]{sihembook}. A Boolean function $f$ over $\mathbb{F}_{2^m}$ is said to be in the class $\mathcal{PS}^{-}$, if $f(0)=0$ and the support of $f$ together with $0$ is the union of $2^{n-1}$ elements of a partial spread of $\mathbb{F}_{2^m}$. Dillon has proved that $deg(f)=\frac{m}{2}$ if $f$ is in $\mathcal{PS}^{-}$ \cite{dillon1}.
 In \cite{youssef}, Youssef and Gong studied the Boolean functions $f$ on the field $\mathbb{F}_{2^m}$, of which the Hamming distance to all functions $\Tr_{1}^{m}(ax^i)+\epsilon$ ($a\in \mathbb{F}_{2^m}$, $\epsilon \in \mathbb{F}_{2}$ and ${\rm gcd}(i,2^m-1)$=1) equals $2^{m-1}\pm 2^{n-1}$. They called these functions \textit{hyper-bent} and showed the existence of such functions in terms of sequences. In \cite{carlet2}, Carlet and Gaborit translated the result into the terminology of Boolean functions.
\begin{prop}{\rm\cite{carlet2,youssef}}\label{propyoussef}
Let $n$ be a positive integer, $q=2^n$ and $\omega$ be a primitive element of $\mathbb{F}_{q^2}$. Let $f$ be a Boolean function on $\mathbb{F}_{q^2}$ such that $f(0)=0$ and $f(\omega^{q+1}x)=f(x)$ for every $x\in \mathbb{F}_{q^2}$. Then $f$ is hyper-bent if and only if the weight of the vector $(f(1),f(\omega),f(\omega^2),\ldots,f(\omega^{q}))$ equals $\frac{q}{2}$.
\end{prop}
In Proposition \ref{propyoussef}, the condition $f(\omega^{q+1}x)=f(x)$ for every $x\in \mathbb{F}_{q^2}$ actually implies that $f$ has the form  \eqref{form2}.
\begin{prop}\label{newprop}
Let $n$ be a positive integer, $q=2^n$ and $\omega$ be a primitive element of $\mathbb{F}_{q^2}$. Let $f$ be a function on $\mathbb{F}_{q^2}$, then $f(x)$ satisfies $f(\omega^{q+1}x)=f(x)$ for every $x\in \mathbb{F}_{q^2}$, if and only if there exists a function $g(x)$ on $\mathbb{F}_{q^2}$ such that
$f(x)=g(x^{q-1})$.
\end{prop}
\begin{proof}
Every function $f$ on $\mathbb{F}_{q^2}$ can be written uniquely as a polynomial
$$f(x)=\sum_{i=0}^{q^2-1}a_{i}x^{i},$$
with $a_{i}\in \mathbb{F}_{q^2}$. Since $f(\omega^{q+1}x)=f(x)$ for every $x\in \mathbb{F}_{q^2}$, i.e., $\sum_{i=0}^{q^2-1}a_{i}\omega^{(q+1)i} x^{i}=\sum_{i=0}^{q^2-1}a_{i}x^{i}$, by the uniqueness of polynomial representation, we have
$$a_{i}(\omega^{(q+1)i}+1)=0.$$
Then $a_{i}=0$ for every $i$ with $(q-1)\nmid i$. Since the sufficiency is obvious, this proves the proposition.
\end{proof}
%In \cite{carlet2}, Carlet and Gaborit showed that the Yossef-Gong functions belonged to or were affine equivalent to class $\mathcal{PS}_{ap}$, a special subclass of $\mathcal{PS}$. They further proved that all functions of class $\mathcal{PS}_{ap}$ were hyper-bent.

An important property of Dillon-like functions, i.e., functions having form \eqref{form2}, is that the Walsh transform can be reduced to a special sum on the unit circle.
\begin{prop}\label{npp}
Let $q=2^n>2$ and $g(x)$ be a function on $\mathbb{F}_{q^2}$ such that $f(x)=g(x^{q-1})$ is a Boolean function. Then
$$W_{f}(\alpha)=\left\{\begin{array}{ll}
(-1)^{g(0)}-\sum_{\lambda\in U}(-1)^{g\left({\lambda}\right)}+(-1)^{g\left(\alpha^{1-q}\right)}\cdot q,& \,{\rm if}\, \alpha\ne 0, \\
(-1)^{g(0)}+(q-1)\cdot\sum_{\lambda \in U}(-1)^{g\left(\overline{\lambda}^2\right)},& \,{\rm if}\, \alpha=0.
\end{array}\right.$$
\end{prop}
\begin{proof}
The proof is direct. For any $\alpha\in \mathbb{F}_{q^2}^{*}$, from the definition of Walsh transform, we have
\begin{eqnarray}
W_{f}(\alpha)&=&\sum_{x\in \mathbb{F}_{q^2}}(-1)^{g(x^{q-1})+\Tr_{1}^{2n}(\alpha x)}\nonumber\\
&=&(-1)^{g(0)}+\sum_{x\in \mathbb{F}_{q^2}^{*}}(-1)^{g(x^{q-1})+{\rm Tr}_{1}^{2n}(\alpha x)}\nonumber\\
&=&(-1)^{g(0)}+\sum_{\lambda\in U}\sum_{y\in \mathbb{F}_{q}^{*}}(-1)^{g\left(\overline{\lambda}^2\right)+{\rm Tr}_{1}^{2n}\left(\alpha \lambda y\right)}\nonumber\\
&=&(-1)^{g(0)}+\sum_{\lambda\in U}(-1)^{g\left(\overline{\lambda}^2\right)}\sum_{y\in \mathbb{F}_{q}^{*}}(-1)^{{\rm Tr}_{1}^{n}((\alpha \lambda+\overline{\alpha} \overline{\lambda}) y)} \nonumber\\
&=&(-1)^{g(0)}-\sum_{\lambda\in U}(-1)^{g\left(\overline{\lambda}^2\right)}+\sum_{\lambda\in U}(-1)^{g\left(\overline{\lambda}^2\right)}\sum_{y\in \mathbb{F}_{q}}(-1)^{{\rm Tr}_{1}^{n}((\alpha \lambda+\overline{\alpha} \overline{\lambda}) y)} \nonumber\\
&=&(-1)^{g(0)}-\sum_{\lambda\in U}(-1)^{g\left(\overline{\lambda}^2\right)}+(-1)^{g\left(\alpha^{1-q}\right)}\cdot q. \nonumber
\end{eqnarray}
where the last equality is obtained because $\alpha\lambda+\overline{\alpha}\overline{\lambda}=0$ means $\overline{\lambda}^2=\alpha^{1-q}$.
Combining the facts
\begin{eqnarray*}
W_{f}(0)&=&(-1)^{g(0)}+\sum_{x\in \mathbb{F}_{q^2}^{*}}(-1)^{g(x^{q-1})}\\
&=&(-1)^{g(0)}+(q-1)\cdot\sum_{\lambda \in U}(-1)^{g\left(\overline{\lambda}^2\right)}
\end{eqnarray*}
and
$$\sum_{\lambda \in U}(-1)^{g\left(\overline{\lambda}^2\right)}=\sum_{\lambda \in U}(-1)^{g({\lambda})},$$
we complete the proof.
\end{proof}
From Proposition \ref{npp}, we easily get the following proposition to determine the bentness of Dillon-like functions, which is a slightly different generalisation of the criterion in \cite{nianli} on the binary case.
\begin{prop}\label{prop1}
Let $q=2^n>2$ and $g(x)$ be a function on $\mathbb{F}_{q^2}$ such that $f(x)=g(x^{q-1})$ is a Boolean function. Then $f$ is bent if and only if
\begin{equation}\label{eq2.2}
\sum_{\lambda\in U}(-1)^{g(\lambda)}=(-1)^{f(0)}.
\end{equation}
\end{prop}
\begin{proof}
Note that $f(0)=g(0)$.
If Eq.\eqref{eq2.2} holds, it is obvious that $W_{f}(\alpha)^2=q^2$ and then $f$ is bent. On the other hand, assume that $f$ is bent. From
$W_{f}(0)=\pm q$, we can see that
\begin{equation}\label{eq2.4}
(-1)^{f(0)}-\sum_{\lambda \in U}(-1)^{g({\lambda})}=q\left(\pm 1-\sum_{\lambda \in U}(-1)^{g({\lambda})}\right).
\end{equation}
Denote $\theta=\sum_{\lambda \in U}(-1)^{g({\lambda})}$ and $l=\sharp\{\,\lambda\in  U: g(\lambda)=0\}$. Then $\theta=l-(q+1-l)=2l-1-q$ is odd. If $f(0)=0$, then Eq.\eqref{eq2.4} becomes
$1-\theta=q(\pm 1-\theta)$. If $1-\theta=q(- 1-\theta)$, of course $\theta\ne-1$. Let $-1-\theta=2k$ for some integer $k\ne0$, we see that $1-\theta=q(- 1-\theta)=2kq$ does not hold for $q>2$, since $-q\le 1-\theta\le q+2$. We get that Eq.\eqref{eq2.4} holds only if $W_{f}(0)=q$ and $\theta=1$. If $f(0)=1$, we have $W_{f}(0)=-q$ and $\theta=-1$ in a similar way.

The proof is finished.
\end{proof}

%We note that the support of the above bent functions are subsets of Desarguesian spreads, which is the collection of $\{u\mathbb{F}_{q}:u \in U\}$.
As seen above,  the condition that $f(\omega^{q+1}x)=f(x)$ holds for all $x\in \mathbb{F}_{q^2}$ in Proposition \ref{newprop} is equivalent to the existence of some function $g$ satisfying $f(x)=g(x^{q-1})$. If $f(0)=0$, then the weight of the vector $(f(1),f(\omega),f(\omega^2),\ldots f(\omega^{q}))$ equals $\frac{q}{2}$ is equivalent to $$1=\sum_{i=0}^{q}(-1)^{f(\omega^i)}=\sum_{i=0}^{q}(-1)^{g(\omega^{i(q-1)})}=\sum_{\lambda\in U}(-1)^{g(\lambda)}=(-1)^{f(0)}.$$
Thus, the condition on the hyper-bentness in Proposition \ref{propyoussef} coincides with the characterization of the bentness in Eq.\eqref{eq2.2}. If $f(0)=1$, a similar discussion can be applied to $f+1$. This implies that Boolean function $f(x)=g(x^{q-1})$ is bent if and only if it is hyper-bent. Since the support of $f(x)=g(x^{q-1})$ is the union of some subspaces from the Desarguesian spread $\{\,u\mathbb{F}_{q}:u \in U\}$, we see that if $f(x)$ is bent, then either $f\in \mathcal{PS}^{-}$ or $f+1\in \mathcal{PS}^{-}$. Proposition \ref{npp} explicitly gives the Walsh transform of a Dillon-like function $f$, and its dual function (if it exists) is easily obtained. Precisely, if Eq.\eqref{eq2.2} is satisfied, i.e., $f(x)=g(x^{q-1})$ is bent, then
$W_{f}(\alpha)=q(-1)^{g(\alpha^{q^2-q})}$
for any $\alpha\in \mathbb{F}_{q^2}$, and then the dual is $\tilde{f}(x)=g(x^{q^2-q})$.

\subsection{Some Relative Exponential Sums}
Proposition \ref{prop1} reduces the bentness of a function with form \eqref{form2} to the sum
\begin{equation*}
\sum_{\lambda\in U}(-1)^{g(\lambda)},
\end{equation*}
which is not easy for general functions $g$. The strategies used in \cite{nianli,yulong} were dividing the exponential sum into some partial sums and then investigating the relationships between these partial sums and some other exponential sums.

The binary Kloosterman sum is defined as
$$\mathcal{K}_{n}(a)=\sum_{x\in \mathbb{F}_{2^n}}(-1)^{{\rm Tr}_{1}^{n}\left(\frac{1}{x}+ax\right)}$$
for each $a\in \mathbb{F}_{2^n}^{*}$.
Thanks to the results of Lachaud and Wolfmann in \cite{lachaud}, we know that the set $\{\,\mathcal{K}_{n}(a):a \in \mathbb{F}_{2^n}\}$ are exactly all multiples of $4$ in the range $[-2^{\frac{n}{2}+1}+1, 2^{\frac{n}{2}+1}+1]$.

The following lemma gives a variant of the Kloosterman sums.
\begin{lem}\label{lem1}
Let $n$ be a positive integer and $q=2^n$. Then for any $a\in \mathbb{F}_{q}^{*}$, we have
\begin{equation}\label{eq2.6}
\sum_{y\in \mathbb{F}_{q},{\rm Tr}_{1}^{n}(\frac{a}{y})=1}(-1)^{{\rm Tr}_{1}^{n}(y)}=-\frac{1}{2}{\mathcal{K}_{n}}(a).
\end{equation}
\end{lem}
\begin{proof}
Note that
\begin{eqnarray*}
  \sum_{y\in \mathbb{F}_{q},{\rm Tr}_{1}^{n}(\frac{a}{y})=1}(-1)^{{\rm Tr}_{1}^{n}(y)} &=&
  \sum_{y\in \mathbb{F}_{q},{\rm Tr}_{1}^{n}(\frac{a}{y})=1,{\rm Tr}_{1}^{n}(y)=0}1- \sum_{y\in \mathbb{F}_{q},{\rm Tr}_{1}^{n}(\frac{a}{y})=1,{\rm Tr}_{1}^{n}(y)=1}1  \\
   &=& 2^{n-1}-2\cdot \sharp D,
\end{eqnarray*}
where $D=\{y\in \mathbb{F}_{q}:{\rm Tr}_{1}^{n}(\frac{a}{y})=1,\,{\rm Tr}_{1}^{n}(y)=1\}$. It is clear that the cardinality $\sharp D$ satisfies
\begin{eqnarray*}
  4\cdot\sharp D &=& \sum_{u_{1},u_{2}\in \mathbb{F}_{2},y\in \mathbb{F}_{q}}(-1)^{u_{1}({\rm Tr}_{1}^{n}(y)+1)+u_{2}({\rm Tr}_{1}^{n}(\frac{a}{y})+1)}  \\
   &=& 2^n-\sum_{y\in \mathbb{F}_{q}}(-1)^{\Tr_{1}^{n}(y)}-\sum_{y\in \mathbb{F}_{q}}(-1)^{\Tr_{1}^{n}(\frac{a}{y})}+\sum_{y\in \mathbb{F}_{q}}(-1)^{{\rm Tr}_{1}^{n}(y+\frac{a}{y})}  \\
   &=& 2^n+\sum_{y\in \mathbb{F}_{q}}(-1)^{{\rm Tr}_{1}^{n}(y+\frac{a}{y})}  \\
   &=& 2^n+\sum_{y\in \mathbb{F}_{q}}(-1)^{{\rm Tr}_{1}^{n}(ay+\frac{1}{y})} \\
   &=& 2^n+\mathcal{K}_{n}(a),
\end{eqnarray*}
which gives Eq.\eqref{eq2.6}.

The proof is finished.
\end{proof}

The following lemma is a crucial result related to the exponential sum of a rational trace function, of which the numerator has degree one and the denominator has degree two. Interestingly, such a sum is also closely related to the Kloosterman sums.

\begin{lem}\label{lem4}
Let $n$ be a positive integer with $q=2^n$, and $A,B,\delta$ be elements in $\mathbb{F}_{q}$ with $A\ne 0$ and $\frac{B^2}{A^2}+\frac{B}{A}+\delta\ne 0$. Denote the exponential sum
$\mathcal{S}=\sum_{x\in \mathbb{F}_{q}}(-1)^{\Tr_{1}^{n}\left(\frac{Ax+B}{x^2+x+\delta}\right)}$ and $C=(AB+A^2\delta)^{\frac{1}{2}}$, we have:

{\rm(1)} $\mathcal{S}=(-1)^{\Tr_{1}^{n}(A)}+(-1)^{\Tr_{1}^{n}(\frac{B+C}{A})}-\mathcal{K}_{n}(B+C)(-1)^{\Tr_{1}^{n}(A)}$ if $\Tr_{1}^{n}(\delta)=1$;

{\rm(2)} $\mathcal{S}=-(-1)^{\Tr_{1}^{n}(A)}-(-1)^{\Tr_{1}^{n}(\frac{B+C}{A})}+\mathcal{K}_{n}(B+C)(-1)^{\Tr_{1}^{n}(A)}+2$ if $\Tr_{1}^{n}(\delta)=0$.
\end{lem}
\begin{proof}
We assume that $\frac{1}{0}$ is $0$. Let $\frac{Ax+B}{x^2+x+\delta}=b$, which is equivalent to
\begin{equation}\label{eq2.7}
   bx^2+(A+b)x+B+b\delta=0
\end{equation}
if $x^2+x+\delta\ne 0$. If $b=0$, then $\frac{Ax+B}{x^2+x+\delta}=0$ is equivalent to $Ax+B=0$ or $x^2+x+\delta=0$. If $b=A$, then Eq.\eqref{eq2.7} has exactly one solution. For $b\ne 0,A$, we know that Eq.\eqref{eq2.7} has two solutions if and only if $\Tr_{1}^{n}(\frac{Bb+b^2\delta}{(A+b)^2})=0$, which reduces to
\begin{eqnarray*}
&&\Tr_{1}^{n}\left(\frac{B(A+b+A)+(A+b+A)^2\delta}{(A+b)^2}\right)\\
&=&\Tr_{1}^{n}\left(\frac{B}{A+b}+\frac{AB+A^2\delta}{(A+b)^2}+\delta\right)\\
&=&\Tr_{1}^{n}\left(\delta+\frac{B+C}{A+b}\right)=0,
\end{eqnarray*}
where $C=(AB+A^2\delta)^{\frac{1}{2}}$. Observe that $(B+C)^2=\left(\frac{B^2}{A^2}+\frac{B}{A}+\delta\right)A^2$, then $B+C\ne0$ due to the assumption. We compute the sum $\mathcal{S}$ according to two cases:

(1) $\Tr_{1}^{n}(\delta)=1$. Then $x^2+x+\delta=0$ has no solutions in $\mathbb{F}_{q}$. Therefore,
\begin{eqnarray*}
  \mathcal{S}&=& 1+(-1)^{\Tr_{1}^{n}(A)}+2\sum_{b\in\mathbb{F}_{q}\setminus\{0,A\},\Tr_{1}^{n}(\frac{B+C}{A+b})=1}(-1)^{\Tr_{1}^{n}(b)} \\
   &=&1+(-1)^{\Tr_{1}^{n}(A)}+2\sum_{b\in\mathbb{F}_{q}\setminus\{0,A\},\Tr_{1}^{n}(\frac{B+C}{b})=1}(-1)^{\Tr_{1}^{n}(b+A)}  \\
   &=&1+(-1)^{\Tr_{1}^{n}(A)}+2(-1)^{\Tr_{1}^{n}(A)}\\
   &&\cdot\left(\sum_{b\in \mathbb{F}_{q},\Tr_{1}^{n}(\frac{B+C}{b})=1}(-1)^{\Tr_{1}^{n}(b)}-(-1)^{\Tr_{1}^{n}(A)}\frac{1-(-1)^{\Tr_{1}^{n}(\frac{B+C}{A})}}{2}\right)\\
   &=&1+(-1)^{\Tr_{1}^{n}(A)}+(-1)^{\Tr_{1}^{n}(A)}\left(-\mathcal{K}_{n}(B+C)-(-1)^{\Tr_{1}^{n}(A)}(1-(-1)^{\Tr_{1}^{n}(\frac{B+C}{A})})\right)\\
   &=&(-1)^{\Tr_{1}^{n}(A)}+(-1)^{\Tr_{1}^{n}(\frac{B+C}{A})}-\mathcal{K}_{n}(B+C)(-1)^{\Tr_{1}^{n}(A)},
\end{eqnarray*}
where the second from the last equality is obtained from Lemma \ref{lem1}.

(2) $\Tr_{1}^{n}(\delta)=0$. In this case, there are three elements $x$ such that $\frac{Ax+B}{x^2+x+\delta}=0$.
Thus,
\begin{eqnarray*}
  \mathcal{S} &=& 3+\sum_{x\in\mathbb{F}_{q}, Ax+B\ne 0,x^2+x+\delta\ne0}(-1)^{\Tr_{1}^{n}\left(\frac{Ax+B}{x^2+x+\delta}\right)}\\
  &=& 3+(-1)^{\Tr_{1}^{n}(A)}+2\sum_{b\in\mathbb{F}_{q}\setminus\{0,A\},\Tr_{1}^{n}(\frac{B+C}{A+b})=0}(-1)^{\Tr_{1}^{n}(b)}\\
  &=& 3+(-1)^{\Tr_{1}^{n}(A)}-2-2(-1)^{\Tr_{1}^{n}(A)}-2\sum_{b\in\mathbb{F}_{q}\setminus\{0,A\},\Tr_{1}^{n}(\frac{B+C}{A+b})=1}(-1)^{\Tr_{1}^{n}(b)}\\
  &=& 1-(-1)^{\Tr_{1}^{n}(A)}-2\sum_{b\in\mathbb{F}_{q}\setminus\{0,A\},\Tr_{1}^{n}(\frac{B+C}{A+b})=1}(-1)^{\Tr_{1}^{n}(b)},
\end{eqnarray*}
where the third equality holds due to the fact
\begin{eqnarray*}
&&\sum_{b\in\mathbb{F}_{q}\setminus\{0,A\},\Tr_{1}^{n}(\frac{B+C}{A+b})=0}(-1)^{\Tr_{1}^{n}(b)}+
\sum_{b\in\mathbb{F}_{q}\setminus\{0,A\},\Tr_{1}^{n}(\frac{B+C}{A+b})=1}(-1)^{\Tr_{1}^{n}(b)}\\
&&=\sum_{b\in\mathbb{F}_{q}\setminus\{0,A\}}(-1)^{\Tr_{1}^{n}(b)}=-1-(-1)^{\Tr_{1}^{n}(A)}.
\end{eqnarray*}
By the computations in case (1), it follows that
\begin{eqnarray*}
  \mathcal{S}&=& 1-(-1)^{\Tr_{1}^{n}(A)}+\mathcal{K}_{n}(B+C)(-1)^{\Tr_{1}^{n}(A)}+1-(-1)^{\Tr_{1}^n{(\frac{B+C}{A})}}\\
  &=&  -(-1)^{\Tr_{1}^{n}(A)}-(-1)^{\Tr_{1}^{n}(\frac{B+C}{A})}+\mathcal{K}_{n}(B+C)(-1)^{\Tr_{1}^{n}(A)}+2.
\end{eqnarray*}
The proof is finished.
\end{proof}
For any $a\in \mathbb{F}_{2^{2n}}$ and $b\in \mathbb{F}_{2^{2n}}^{*}$, we introduce two exponential sums on the unit circle $U$ as
\begin{equation}\label{eq3.1}
  \xi(a,b)=\sum_{\lambda \in U}(-1)^{\Tr_{1}^{2n}\left(\frac{a}{\lambda+b}\right)} ~\hbox{and}~ M(a,b)=\sum_{\lambda \in U\setminus\{1\}}(-1)^{\Tr_{1}^{2n}\left(\frac{a}{\lambda+b}\right)}.
\end{equation}
Obviously $\xi(0,b)=q+1$. From Proposition \ref{npp}, we know that $\xi(a,b)$ is connected with the Walsh transform of these functions in set $\mathcal{H}$.

By noting that $$\xi(a,b)=(-1)^{\Tr_{1}^{2n}(\frac{a}{1+b})}+M(a,b)$$
and Lemma \ref{lem4}, we can establish a relation between $\xi(a,b)$ and Kloosterman sum through $M(a,b)$, which is important for all characterizations in Section \ref{mainresults}.
\begin{lem}\label{lem5}
Let $q=2^n$ and $a,b\in \mathbb{F}_{q^2}^{*}$, we have
\begin{equation}\label{eq3.2}
\xi(a,b)=\left\{\begin{array}{ll}
     1+(-1)^{\Tr_{1}^{n}(a\overline{a})}q, & {\rm if}\, b\in U ~\hbox{and}~ \Tr_{n}^{2n}(a\overline{b})=0 , \\
     1,    &{\rm if}\, b\in U ~\hbox{and}~ \Tr_{n}^{2n}(a\overline{b})\ne0 ,  \\
     \left(1-\mathcal{K}_{n}\left(\frac{a\overline{a}}{1+b^2\overline{b}^2}\right)\right)(-1)^{\Tr_{1}^{2n}\left(\frac{\overline{a}b}{1+b\overline{b}}\right)},    & {\rm if}\,b\in \mathbb{F}_{q^2}^{*}\setminus U.
     \end{array}\right.
\end{equation}
\end{lem}
\begin{proof}
Firstly we observe that, if $a\in \mathbb{F}_{q}^{*}$, then Eq.\eqref{eq3.2} becomes
\begin{equation}\label{eq3.4}
\xi(a,b)=\left\{\begin{array}{ll}
     1+(-1)^{\Tr_{1}^{n}(a)}q, & {\rm if}\,b=1,\\
     1,    & {\rm if}\,b\in U\setminus\{1\},  \\
     \left(1-\mathcal{K}_{n}\left(\frac{a}{1+b\overline{b}}\right)\right)(-1)^{\Tr_{1}^{n}\left(\frac{a(b+\overline{b})}{1+b\overline{b}}\right)},    & {\rm if}\,b\in \mathbb{F}_{q^2}^{*}\setminus U.
     \end{array}\right.
\end{equation}
If $a\notin \mathbb{F}_{q}$, we have the following relation
\begin{eqnarray*}
  \xi(a,b) &=& \sum_{\lambda\in U}(-1)^{\Tr_{1}^{2n}\left(\frac{a}{\lambda+b}\right)}  \\
   &=& \sum_{\lambda\in U}(-1)^{\Tr_{1}^{2n}\left(\frac{a^2}{\lambda^2+b^2}\right)}  \\
   &=& \sum_{\lambda\in U}(-1)^{\Tr_{1}^{2n}\left(\frac{a^{1-q}a^{1+q}}{\lambda+b^2}\right)}  \\
   &=& \sum_{\lambda\in U}(-1)^{\Tr_{1}^{2n}\left(\frac{a^{1-q}a^{1+q}}{a^{1-q}\lambda+b^2}\right)}  \\
   &=& \sum_{\lambda\in U}(-1)^{\Tr_{1}^{2n}\left(\frac{a\overline{a}}{\lambda+a^{q-1}b^2}\right)}  \\
   &=& \xi(a\overline{a},a^{q-1}b^2).
\end{eqnarray*}
%\begin{equation}\label{eq3.4}
%\xi(a,b)=\left\{\begin{array}{ll}
%     1+(-1)^{\Tr_{1}^{n}(a)}q, & b=1,\\
%     1,    & b\in U\setminus\{1\},  \\
%     \left(1-\mathcal{K}_{n}\left(\frac{a}{1+b\overline{b}}\right)\right)(-1)^{\Tr_{1}^{n}\left(\frac{a(b+\overline{b})}{1+b\overline{b}}\right)},    & b\in %\mathbb{F}_{q^2}^{*}\setminus U.
%     \end{array}\right.
%\end{equation}
If $b\in U$, then $a^{q-1}b^2=1$ is equivalent to $\Tr_{n}^{2n}(a\overline{b})=0$. For every $a\in \mathbb{F}_{q^2}^{*}$, we have that $b\in U$ if and only if $a^{q-1}b^2\in U$. Then Eq.\eqref{eq3.2} follows from Eq.\eqref{eq3.4}. Thus it suffices for us to prove Eq.\eqref{eq3.4}.

 By the relation between $\xi(a,b)$ and $M(a,b)$, we proceed the proof by the computation of $M(a,b)$, according to the cases $b\in \mathbb{F}_{q}^{*}$ and $b\notin \mathbb{F}_{q}$.

Case $b\in \mathbb{F}_{q}^{*}$. From $\lambda\in U\setminus\{1\}$ and
$${\rm Tr}_{1}^{2n}\left(\frac{a}{\lambda+b}\right)={\rm Tr}_{1}^{n}\left(\frac{a}{\lambda+b}+\frac{a}{\overline{\lambda}+{b}}\right)={\rm Tr}_{1}^{n}\left(\frac{a(\lambda+\overline{\lambda})}{1+b^2+b(\lambda+\overline{\lambda})}\right),$$
it is easy to see that $M(a,1)=(-1)^{\Tr_{1}^{n}(a)}q$. Then, assume $b\ne 1$ and denote $\alpha=\lambda+\overline{\lambda}$. From Lemma \ref{lem2.1}(1), there are exactly two elements $\lambda\in U\setminus\{1\}$ corresponding to one $\alpha$ if and only if $\Tr_{1}^{n}(\alpha^{-1})=1$, then
\begin{eqnarray*}
M(a,b)&=&2\sum_{{\rm Tr}_{1}^{n}(\alpha^{-1})=1}(-1)^{{\rm Tr}_{1}^{n}\left({\frac{a\alpha}{1+b^2+b\alpha}}\right)}\,\left({\rm letting}\, \beta=\frac{a\alpha}{1+b^2+b\alpha}\right)\\
&=&2\sum_{{\rm Tr}_{1}^{n}\left(\frac{a+b\beta}{(1+b^2)\beta}\right)=1}(-1)^{{\rm Tr}_{1}^{n}\left(\beta\right)}\\
&=&2\sum_{{\rm Tr}_{1}^{n}\left(\frac{a}{\beta(1+b^2)}\right)=1}(-1)^{{\rm Tr}_{1}^{n}(\beta)}\\
&=&-\mathcal{K}_{n}\left(\frac{a}{1+b^2}\right)
\end{eqnarray*}
due to Lemma \ref{lem1}. Then Eq.\eqref{eq3.4} holds.

Case $b\notin \mathbb{F}_{q}$. The bijectivity between $U\setminus\{1\}$ and $\mathbb{F}_{q}$ in Lemma \ref{lem0} gives
\begin{eqnarray}
M(a,b)&=&\sum_{z\in \mathbb{F}_{q}}(-1)^{\Tr_{1}^{2n}\left(\frac{a}{\frac{z+b}{z+\overline{b}}+b}\right)}\nonumber\\
&=&\sum_{z\in\mathbb{F}_{q}}(-1)^{\Tr_{1}^{2n}\left(\frac{a(z+\overline{b})}{(1+b)z+b+b\overline{b}}\right)}\nonumber\\
&=&\sum_{z\in\mathbb{F}_{q}}(-1)^{\Tr_{1}^{2n}\left(\frac{\frac{a}{1+b}z+\frac{a\overline{b}}{1+b}}{z+\frac{b(1+\overline{b})}{1+b}}\right)}\label{eq3.3},
\end{eqnarray}
which admits further computations of $M(a,b)$ according to two subcases $b\in U$ or $b\notin U$.
%Note that $\frac{b(1+\overline{b})}{1+b}\in \mathbb{F}_{q}$ if and only if $b\overline{b}=1$, which means that $\frac{b(1+\overline{b})}{1+b}=1$.

(1) $b\in U$. Note that $b\ne 1$ due to the assumption $b\notin \mathbb{F}_{q}$. We directly have
\begin{eqnarray*}
M(a,b)&=&\sum_{z\in\mathbb{F}_{q}}(-1)^{\Tr_{1}^{2n}\left(\frac{\frac{a}{1+b}z+\frac{a\overline{b}}{1+b}}{z+1}\right)}\\
&=&1+\sum_{z\ne 1}(-1)^{\Tr_{1}^{2n}\left(\frac{\frac{a}{1+b}z+\frac{a\overline{b}}{1+b}}{z+1}\right)}\\
&=&1+\sum_{z\ne 1}(-1)^{\Tr_{1}^{2n}\left(\frac{\frac{a}{1+b}(z+1)+\frac{a\overline{b}+a}{1+b}}{z+1}\right)}\\
&=&1+(-1)^{\Tr_{1}^{2n}\left(\frac{a}{1+b}\right)}\sum_{z\ne 1}(-1)^{\Tr_{1}^{2n}\left(\frac{a(1+\overline{b})}{(z+1)(1+b)}\right)}\\
&=&1+(-1)^{\Tr_{1}^{2n}\left(\frac{a}{1+b}\right)}\sum_{z\ne 1}(-1)^{\Tr_{1}^{n}\left(\frac{a(b+\overline{b})}{z+1}\right)}\\
&=&1-(-1)^{\Tr_{1}^{2n}\left(\frac{a}{1+b}\right)}.
\end{eqnarray*}

(2) $b\notin U$. In this case, Eq.\eqref{eq3.3} becomes
\begin{eqnarray*}
M(a,b)&=&\sum_{z\in\mathbb{F}_{q}}(-1)^{\Tr_{1}^{2n}\left(\frac{a}{1+b}+\frac{\frac{a\overline{b}}{1+b}+\frac{ab(1+\overline{b})}{1+b^2}}{z+\frac{b(1+\overline{b})}{1+b}}\right)}\\
&=&(-1)^{\Tr_{1}^{2n}(\frac{a}{1+b})}\sum_{z\in\mathbb{F}_{q}}(-1)^{\Tr_{1}^{n}\left(\frac{\frac{a\overline{b}}{1+b}+\frac{ab(1+\overline{b})}{1+b^2}}{z+\frac{b(1+\overline{b})}{1+b}}
+\frac{\frac{ab}{1+\overline{b}}+\frac{a\overline{b}(1+{b})}{1+\overline{b}^2}}{z+\frac{\overline{b}(1+{b})}{1+\overline{b}}}\right)}\\
&=&(-1)^{\Tr_{1}^{2n}(\frac{a}{1+b})}\sum_{z\in\mathbb{F}_{q}}(-1)^{\Tr_{1}^{n}\left(\frac{\frac{a(b+\overline{b})}{1+b^2}}{z+\frac{b(1+\overline{b})}{1+b}}
+\frac{\frac{a(b+\overline{b})}{1+\overline{b}^2}}{z+\frac{\overline{b}(1+{b})}{1+\overline{b}}}\right)}.
\end{eqnarray*}
Note that
\begin{eqnarray*}\frac{\frac{a(b+\overline{b})}{1+b^2}}{z+\frac{b(1+\overline{b})}{1+b}}
+\frac{\frac{a(b+\overline{b})}{1+\overline{b}^2}}{z+\frac{\overline{b}(1+{b})}{1+\overline{b}}}&=&
\frac{\left(\frac{1}{1+\overline{b}^2}+\frac{1}{1+b^2}\right)a(b+\overline{b})z+\frac{a(b+\overline{b})^2}{(1+b)(1+\overline{b})}}
{\left(z+\frac{\overline{b}(1+{b})}{1+\overline{b}}\right)\left(z+\frac{b(1+\overline{b})}{1+b}\right)}\\
&=&\frac{A_{1}z+B_{1}}{z^2+A_{2}z+B_{2}},
\end{eqnarray*}
where
$$\left\{\begin{array}{ll}
   A_{1}&= \frac{a(b+\overline{b})^3}{(1+b^2)(1+\overline{b}^2)},  \\
   B_{1}&= \frac{a(b+\overline{b})^2}{(1+b)(1+\overline{b})}, \\
   A_{2}&= \frac{(b+\overline{b})(1+b\overline{b})}{(1+b)(1+\overline{b})},\\
   B_{2}&=  b\overline{b}.
\end{array}\right.$$
Since $A_{2}\ne 0$, we further obtain
\begin{eqnarray}
  M(a,b) &=& (-1)^{\Tr_{1}^{2n}\left(\frac{a}{1+b}\right)}\sum_{z\in\mathbb{F}_{q}}(-1)^{\Tr_{1}^{n}\left(\frac{A_{1}z+B_{1}}{z^2+A_{2}z+B_{2}}\right)}  ~~~~~~(\hbox{substituting}~ z ~\hbox{with} ~A_{2}z)\nonumber\\
   &=& (-1)^{\Tr_{1}^{2n}\left(\frac{a}{1+b}\right)}\sum_{z\in\mathbb{F}_{q}}(-1)^{\Tr_{1}^{n}\left(\frac{Az+B}{z^2+z+\delta}\right)} , \label{xeq2}
\end{eqnarray}
where
$$\left\{
\begin{array}{cl}
A&=\frac{A_{1}}{A_{2}}=\frac{a(b+\overline{b})^2}{(1+b\overline{b})(1+b)(1+\overline{b})},\\
B&=\frac{B_{1}}{A_{2}^2}=\frac{a(1+b)(1+\overline{b})}{(1+b\overline{b})^2}\\
\delta&=\frac{B_{2}}{A_{2}^2}=\frac{b\overline{b}(1+b^2)(1+\overline{b}^2)}{(b+\overline{b})^2(1+b\overline{b})^2}.
\end{array}\right.$$
Observe that
\begin{eqnarray*}
\delta&=&
\frac{b\overline{b}\left((1+b\overline{b})^2+(b+\overline{b})^2\right)}{(b+\overline{b})^2(1+b\overline{b})^2}\\
&=&\frac{b\overline{b}}{(b+\overline{b})^2}+\frac{1}{1+b\overline{b}}+\frac{1}{(1+b\overline{b})^2}.
\end{eqnarray*}
By the assumption $b\notin \mathbb{F}_{q}$, and that $b$ and $\overline{b}$ are two solutions of the quadratic equation $x^2+(b+\overline{b})x+b\overline{b}=0$, we have $\Tr_{1}^{n}(\delta)=\Tr_{1}^{n}\left(\frac{b\overline{b}}{(b+\overline{b})^2}\right)=1$. Then according to Lemma \ref{lem4}, we can rewrite Eq.\eqref{xeq2} as
\begin{eqnarray*}
  (-1)^{\Tr_{1}^{2n}(\frac{a}{1+b})}\cdot M(a,b) &=& (-1)^{\Tr_{1}^{n}(A)}+(-1)^{\Tr_{1}^{n}\left(\frac{B+C}{A}\right)}-\mathcal{K}_{n}(B+C)(-1)^{\Tr_{1}^{n}(A)},
\end{eqnarray*}
where $C=(AB+A^2\delta)^{\frac{1}{2}}$. It is straightforward that
\begin{eqnarray*}
&&AB+A^2\delta\\
&=&\frac{a(b+\overline{b})^2}{(1+b\overline{b})(1+b)(1+\overline{b})}\frac{a(1+b)(1+\overline{b})}{(1+b\overline{b})^2}
  +\frac{a^2(b+\overline{b})^4}{(1+b\overline{b})^2(1+b)^2(1+\overline{b})^2}\frac{b\overline{b}(1+b)^2(1+\overline{b})^2}{(b+\overline{b})^2(1+b\overline{b})^2}\\
&=&\frac{a^2(b+\overline{b})^2}{(1+b\overline{b})^3}+\frac{a^2(b+\overline{b})^2b\overline{b}}{(1+b\overline{b})^4}\\
&=&\frac{a^2(b+\overline{b})^2}{(1+b\overline{b})^3}
\left(1+\frac{b\overline{b}}{1+b\overline{b}}\right)=\frac{a^2(b+\overline{b})^2}{(1+b\overline{b})^4},
\end{eqnarray*}
which gives
$$C=\frac{a(b+\overline{b})}{(1+b\overline{b})^2}$$
and then
$$B+C=\frac{a(b+\overline{b})+a(1+b\overline{b}+b+\overline{b})}{(1+b\overline{b})^2}=\frac{a}{1+b\overline{b}}.$$
Another important fact is $\Tr_{1}^{n}\left(\frac{B+C}{A}\right)=1$, due to
\begin{eqnarray*}
              \frac{B+C}{A} &=&\frac{a}{1+b\overline{b}}\frac{(1+b\overline{b})(1+b)(1+\overline{b})}{a(b+\overline{b})^2}  \\
               &=&\frac{(1+b)(1+\overline{b})}{(b+\overline{b})^2}  \\
               &=& \frac{b\overline{b}}{(b+\overline{b})^2}+\frac{1}{b+\overline{b}}+\frac{1}{(b+\overline{b})^2}.
            \end{eqnarray*}
Together with
\begin{eqnarray*}
  A+\frac{a}{1+b}+\frac{a}{1+\overline{b}} &=&\frac{a(b+\overline{b})^2}{(1+b\overline{b})(1+b)(1+\overline{b})}+\frac{a(b+\overline{b})}{(1+b)(1+\overline{b})}  \\
   &=& \frac{a(b+\overline{b})}{(1+b)(1+\overline{b})}\left(\frac{b+\overline{b}}{1+b\overline{b}}+1\right) \\
   &=& \frac{a(b+\overline{b})}{1+b\overline{b}},
\end{eqnarray*}
we then easily obtain
$$M(a,b)=\left(1-\mathcal{K}_{n}\left(\frac{a}{1+b\overline{b}}\right)\right)(-1)^{\Tr_{1}^{n}\left(\frac{a(b+\overline{b})}{1+b\overline{b}}\right)}-(-1)^{\Tr_{1}^{2n}\left(\frac{a}{1+b}\right)}.$$
Hence, the above discussion on $M(a,b)$ gives $\xi(a,b)$ in Eq.\eqref{eq3.4}.

The proof is finished.
\end{proof}

\section{Main Results}\label{mainresults}
In this section, based on the family $\mathcal{H}=\{\Tr_{1}^{m}\left(\frac{a}{x^{q-1}+b}\right):a,b\in \mathbb{F}_{q^2}^{*}\}$, we consider the bentness of functions $h=F(f_{1},f_{2},\ldots,f_{t})$ in \eqref{form3} on $\mathbb{F}_{q^2}$ with $q=2^n$ and $f_{i}\in \mathcal{H}$.

%\subsection{A generic method on constructing bent functions}
Recall that $h$ is also Dillon-like, then by Proposition \ref{prop1} and letting
$$\mathcal{T}=\sum_{\lambda\in U}(-1)^{F\left(\Tr_{1}^{m}\left(\frac{a_{1}}{\lambda+b}\right),\Tr_{1}^{m}\left(\frac{a_{2}}{\lambda+b}\right),\ldots,\Tr_{1}^{m}\left(\frac{a_{t}}{\lambda+b}\right)\right)} ,$$
we know that $h$ is bent if and only if
\begin{equation*}
  \mathcal{T}=(-1)^{F\left(\Tr_{1}^{m}\left(\frac{a_{1}}{b}\right),\Tr_{1}^{m}\left(\frac{a_{2}}{b}\right),\ldots,\Tr_{1}^{m}\left(\frac{a_{t}}{b}\right)\right)}.
\end{equation*}
Denoting the Walsh transform of $F(X_{1},X_{2},\ldots,X_{t})$ by $\{W_{F}(\alpha):\alpha\in \mathbb{F}_{2}^{t}\}$. It follows from
the inverse Walsh transform of $F$ that
$$2^t\cdot(-1)^{F(X_{1},X_{2},\ldots,X_{t})}=\sum_{\alpha\in \mathbb{F}_{2}^{t}}W_{F}(\alpha)(-1)^{\sum_{i=1}^{t}s_{i}X_{i}},$$
where $\alpha=(s_{1},s_{2},\ldots,s_{t})$. By replacing the variable $X_{i}$ with $\Tr_{1}^{m}\left(\frac{a_{i}}{\lambda+b}\right)$ for $i=1,2,\ldots,t$, we have
\begin{eqnarray*}
&&2^t\cdot(-1)^{F\left(\Tr_{1}^{m}\left(\frac{a_{1}}{\lambda+b}\right),\Tr_{1}^{m}\left(\frac{a_{2}}{\lambda+b}\right),\ldots,\Tr_{1}^{m}\left(\frac{a_{t}}{\lambda+b}\right)\right)}\\
&=&\sum_{\alpha\in \mathbb{F}_{2}^{t}}W_{F}(\alpha)(-1)^{\sum_{i=1}^{t}s_{i}\Tr_{1}^{m}\left(\frac{a_{i}}{\lambda+b}\right)}
\end{eqnarray*}
and then
\begin{eqnarray*}
&&\sum_{\lambda\in U}2^t\cdot(-1)^{F\left(\Tr_{1}^{m}\left(\frac{a_{1}}{\lambda+b}\right),\Tr_{1}^{m}\left(\frac{a_{2}}{\lambda+b}\right),\ldots,\Tr_{1}^{m}\left(\frac{a_{t}}{\lambda+b}\right)\right)}\\
&=&\sum_{\lambda\in U}\sum_{\alpha\in \mathbb{F}_{2}^{t}}W_{F}(\alpha)(-1)^{\sum_{i=1}^{t}s_{i}\Tr_{1}^{m}\left(\frac{a_{i}}{\lambda+b}\right)}\\
&=&\sum_{\alpha\in \mathbb{F}_{2}^{t}}W_{F}(\alpha)\sum_{\lambda\in U}(-1)^{\Tr_{1}^{m}\left(\frac{\sum_{i=1}^{t}s_{i}a_{i}}{\lambda+b}\right)}\\
&=&\sum_{\alpha\in \mathbb{F}_{2}^{t}}W_{F}(\alpha)\xi\left(\sum_{i=1}^{t}s_{i}a_{i},b\right),
\end{eqnarray*}
where $\xi(\cdot,b)$ is defined as in Eq.\eqref{eq3.1}.
Therefore, to characterize the bentness of $h$, it is equivalent to find these parameters $a_{1},a_{2},\ldots,a_{t}$ and $b$ such that
\begin{equation}\label{xxeq}
2^t (-1)^{F\left(\Tr_{1}^{m}\left(\frac{a_{1}}{b}\right),\Tr_{1}^{m}\left(\frac{a_{2}}{b}\right),\ldots,\Tr_{1}^{m}\left(\frac{a_{t}}{b}\right)\right)}
=\sum_{\alpha\in \mathbb{F}_{2}^{t}}W_{F}(\alpha)\xi\left(\sum_{i=1}^{t}s_{i}a_{i},b\right).
\end{equation}

%\subsection{Characterizations on three classes of bent functions}
Let  $\omega$ be a primitive element of $\mathbb{F}_{q^2}$, then $a_{1}=w^{i(q-1)}a$ with some integer $i$ and $a\in \mathbb{F}_{q}^{*}$. By substituting $x$ with $w^ix$, then $f_{1}(x)=\Tr_{1}^{2n}\left(\frac{a_{1}}{x^{q-1}+b}\right)$ becomes
 $\Tr_{1}^{2n}\left(\frac{a}{x^{q-1}+\frac{b}{w^{i(q-1)}}}\right)$. Therefore, up to affine equivalence of functions, we can always assume $a_{1}\in \mathbb{F}_{q}^{*}$ in \eqref{form3}.

By preceding Lemma \ref{lem5}, we can now describe our three classes of Dillon-like bent functions $h$ in \eqref{form3} with $F(X_{1},X_{2},\ldots,X_{t})\in \{X_{1}, X_{1}X_{2}, X_{1}X_{2}+X_{1}X_{3}+X_{2}X_{3}\}$. The first theorem is on the case $F_{1}=X_{1}$.
\begin{thm}\label{thm1}
Let $n$ be a positive integer, $q=2^n$ and $a_{1}\in \mathbb{F}_{q}^{*},b\in \mathbb{F}_{q^2}^{*}$. Then $h_{1}(x)=F_{1}(f_{1})={\rm Tr}_{1}^{2n}\left(\frac{a_{1}}{x^{q-1}+b}\right)$ is bent if and only if one of the following conditions is satisfied:

{\rm (1)} $b\in U\setminus\{1\}$, ${\rm Tr}_{1}^{2n}(\frac{a_{1}}{b})=0$;

{\rm (2)} $b \in{\mathbb{F}_{q^2}^{*}}\setminus U$, ${\mathcal{K}}_{n}\left(\frac{a_{1}}{1+b\overline{b}}\right)=0$ and ${\rm Tr}_{1}^{n}\left(\frac{a_{1}(b+\overline{b})}{(1+b\overline{b})b\overline{b}}\right)=0$.
\end{thm}
\begin{proof}
By $t=1$ in Eq.\eqref{xxeq}, we know that $h_{1}(x)$ is bent if and only if
\begin{equation}\label{eq3.5}
\xi(a_{1},b)=(-1)^{\Tr_{1}^{2n}(\frac{a_{1}}{b})},
\end{equation}
due to $W_{F_{1}}(0)=0$ and $W_{F_{1}}(1)=2$. If $b=1$, we know $\xi(a_{1},1)=1+(-1)^{\Tr_{1}^{n}(a)}q$ from Eq.\eqref{eq3.4}, then $h_{1}(x)$ is not bent.
If $b\in U\setminus\{1\}$, then $\xi(a_{1},b)=1$ and Eq.\eqref{eq3.5} holds if and only if $\Tr_{1}^{2n}(\frac{a_{1}}{b})=0$.
%If $b\in \mathbb{F}_{q}^{*}\setminus\{1\}$, then
%$\xi(a,b)=1-\mathcal{K}_{n}(\frac{a}{1+b^2})$ and $\Tr_{1}^{2n}(\frac{a}{b})=0$. It is easy to see that $h_{1}(x)$ is bent if and only if %$\mathcal{K}_{n}(\frac{a}{1+b^2})=0$.
When $b\notin U$, we have $\xi(a_{1},b)=\left(1-\mathcal{K}_{n}\left(\frac{a_{1}\overline{a_{1}}}{1+b^2\overline{b}^2}\right)\right)(-1)^{\Tr_{1}^{2n}\left(\frac{\overline{a_{1}}b}{1+b\overline{b}}\right)}$ again by Lemma \ref{lem5}. Then Eq.\eqref{eq3.5} becomes
$$1-\mathcal{K}_{n}\left(\frac{a_{1}}{1+b\overline{b}}\right)=(-1)^{\Tr_{1}^{n}\left(\frac{a_{1}(b+\overline{b})}{1+b\overline{b}}\right)+\Tr_{1}^{2n}(\frac{a_{1}}{b})}
=(-1)^{\Tr_{1}^{n}\left(\frac{a_{1}(b+\overline{b})}{(1+b\overline{b})b\overline{b}}\right)}.$$
By the fact that the value of Kloosterman sum is a multiple of $4$, we get that
Eq.\eqref{eq3.5} holds if and only if $\mathcal{K}_{n}\left(\frac{a_{1}}{1+b\overline{b}}\right)=0$ and $\Tr_{1}^{n}\left(\frac{a_{1}(b+\overline{b})}{(1+b\overline{b})b\overline{b}}\right)=0$.

The proof is completed.
\end{proof}

\begin{example}
Let $q=2^3$ and $w\in \mathbb{F}_{q^2}$ be a primitive element with minimal polynomial $m(x)=x^6+x^4+x^3+x+1\in \mathbb{F}_{2}[x]$. Choose $a_{1}=1$, $b_{1}=w^7$ and $b_{2}=w$, it can be verified that $(a_{1},b_{1})$ satisfies condition (1) and $(a_{1},b_{2})$ satisfies (2), then both functions $\Tr_{1}^{6}(\frac{1}{x^{7}+w^7})$ and $\Tr_{1}^{6}(\frac{1}{x^{7}+w})$ are bent.
\end{example}

%Tang et.al gave a generic construction of bent functions in \cite{tangchunming} of the form
%$$f(x)=g(x)+F(f_{1},f_{2},\ldots,f_{t}),$$
%where $g(x)$ were some known bent functions on $\mathbb{F}_{2^m}$, $F(X_{1},X_{2},\ldots,X_{t})$ were reduced polynomials in %$\mathbb{F}_{2}[X_{1},X_{2},\ldots,X_{t}]$ and $f_{i}=\Tr_{1}^{m}(u_{i}x)$ with $u_{i}\in \mathbb{F}_{2^{m}}$, $i=1,2,\ldots,t$.

In what follows we characterize the bentness of $h$ for $F_{2}=X_{1}X_{2}$.
%we find that
%\begin{equation}\label{eq3.7}
%h_{2}(x)=\Tr_{1}^{2n}\left(\frac{a_{1}}{x^{q-1}+b}\right)\Tr_{1}^{2n}\left(\frac{a_{2}}{x^{q-1}+b}\right),
%\end{equation}
%can be bent for certain elements $a_{1},a_{2}$ and $b$ belonging to $\mathbb{F}_{q^2}^{*}$.
\begin{thm}{\label{thmh1}}
Let $q=2^n$ with positive integer $n\ge6$, $a_{1}\in \mathbb{F}_{q}^{*}$, $a_{2},b\in \mathbb{F}_{q^2}^{*}$ and $a_{1}\ne a_{2}$. Then $h_{2}(x)=F_{2}(f_{1},f_{2})=f_{1}f_{2}$ is bent if and only if one of the following conditions is satisfied:

{\rm(1)} $b=1$, $\Tr_{1}^{n}(a_{1})=1$ and $a_{2}\in \mathbb{F}_{q^2}^{*}\setminus\mathbb{F}_{q}^{*}$;

{\rm(2)} $b\in U\setminus\{1\}$, then $\Tr_{n}^{2n}(a_{2}\overline{b})=0$, $\Tr_{1}^{n}(a_{2}\overline{a_{2}})=1$ and $h_{2}(0)=0$;

{\rm(3)} $b\in U\setminus\{1\}$, $\Tr_{n}^{2n}((a_{1}+a_{2})\overline{b})=0$, $\Tr_{1}^{n}((a_{1}+a_{2})(a_{1}+\overline{a_{2}}))=0$ and $h_{2}(0)=0$.
\end{thm}
\begin{proof}
For $\alpha=(s_{1},s_{2})\in \mathbb{F}_{2}^{2}$, it is clear that
$W_{F_{2}}(s_{1},s_{2})=2(-1)^{s_{1}s_{2}}$.
Then Eq.\eqref{xxeq} becomes
\begin{equation}\label{eq3.9}
q+1+\xi(a_{1},b)+\xi(a_{2},b)-\xi(a_{1}+a_{2},b)=2(-1)^{h_{2}(0)},
\end{equation}
where $\xi(\cdot,b)$ is as defined in Eq.\eqref{eq3.1}. For convenience, we denote the left-hand side of Eq.\eqref{eq3.9} by $\Theta$. The proof is also proceeded according to values of $b$.

Case $b=1$. By Eq.\eqref{eq3.2}, we have
\begin{equation}\label{eq3.10}
\xi(a,1)=\left\{\begin{array}{ll}
1+(-1)^{\Tr_{1}^{n}(a)}q, & {\rm if}\,a\in\mathbb{F}_{q},\\
1, & {\rm if}\,a\notin \mathbb{F}_{q}.
\end{array}\right.
\end{equation}
Note that $a_{1}\in \mathbb{F}_{q}$, then $a_{2}\in \mathbb{F}_{q}$ if and only if $a_{1}+a_{2}\in \mathbb{F}_{q}$. If $a_{2}\notin \mathbb{F}_{q}$, then
$\xi(a_{2},1)=\xi(a_{1}+a_{2},1)=1$ and then
\begin{equation*}
   \Theta=q+1+1+(-1)^{\Tr_{1}^{n}(a_{1})}q+1-1
   = 2+q(1+(-1)^{\Tr_{1}^{n}(a_{1})}).
\end{equation*}
Since $h_{2}(0)=\Tr_{1}^{2n}(a_{1})\Tr_{1}^{2n}(a_{2})=0$, we deduce that $h_{2}$ is bent, i.e., Eq.\eqref{eq3.9} holds, if and only if $\Tr_{1}^{n}(a_{1})=1$. If $a_{2}\in \mathbb{F}_{q}$, then $\xi(a_{2},1)=1+(-1)^{\Tr_{1}^{n}(a_{2})}q$, $\xi(a_{1}+a_{2},1)=1+(-1)^{\Tr_{1}^{n}(a_{1}+a_{2})}q$ and
\begin{eqnarray*}
  \Theta&=&q+1+1+(-1)^{\Tr_{1}^{n}(a_{1})}q+1+(-1)^{\Tr_{1}^{n}(a_{2})}q-1-(-1)^{\Tr_{1}^{n}(a_{1}+a_{2})}q  \\
   &=& 2+q(1+(-1)^{\Tr_{1}^{n}(a_{1})}+(-1)^{\Tr_{1}^{n}(a_{2})}-(-1)^{\Tr_{1}^{n}(a_{1}+a_{2})})\\
   &=& 2+2q(-1)^{\Tr_{1}^{n}(a_{1})\Tr_{1}^{n}(a_{2})}.
\end{eqnarray*}
Thus Eq.\eqref{eq3.9} does not hold, which implies that $h_{2}$ is not bent if $a_{2}\in \mathbb{F}_{q}$. Then condition (1) follows.

Case $b\in U\setminus\{1\}$. Since
\begin{equation}\label{eq3.11}
\xi(a,b)=\left\{\begin{array}{ll}
1+(-1)^{\Tr_{1}^{n}(a\overline{a})}q, & {\rm if}\,\Tr_{n}^{2n}(a\overline{b})=0,\\
1, & {\rm if}\,\Tr_{n}^{2n}(a\overline{b})\ne 0,
\end{array}\right.
\end{equation}
we firstly have $\xi(a_{1},b)=1$ due to $a_{1}\in \mathbb{F}_{q}$ and $\Tr_{n}^{2n}(a_{1}\overline{b})=a_{1}(b+\overline{b})\ne0$. If $\Tr_{n}^{2n}(a_{2}\overline{b})\ne 0$ and $\Tr_{n}^{2n}((a_{1}+a_{2})\overline{b})\ne 0$ hold at the same time, then
$\Theta=q+1+1+1-1=q+2$, which implies that $h_{2}(x)$ is not bent. If $\Tr_{n}^{2n}(a_{2}\overline{b})=0$ and $\Tr_{n}^{2n}((a_{1}+a_{2})\overline{b})=0$, then $\Tr_{n}^{2n}(a_{1}\overline{b})=a_{1}(b+\overline{b})=0$. This contradicts to the assumption $b\in U\setminus\{1\}$. Thus, if $\Tr_{n}^{2n}(a_{2}\overline{b})=0$, then $\Tr_{n}^{2n}((a_{1}+a_{2})\overline{b})\ne0$, and then
\begin{eqnarray*}
   \Theta&=& q+1+1+1+(-1)^{\Tr_{1}^{n}(a_{2}\overline{a_{2}})}q-1 \\
   &=&2+q(1+(-1)^{\Tr_{1}^{n}(a_{2}\overline{a_{2}})}).
\end{eqnarray*}
Therefore, Eq.\eqref{eq3.9} holds if and only if $\Tr_{1}(a_{2}\overline{a_{2}})=1$ and $h_{2}(0)=0$. On the other hand, if $\Tr_{n}^{2n}((a_{1}+a_{2})\overline{b})=0$, then $\Tr_{n}^{2n}(a_{2}\overline{b})=a_{1}(b+\overline{b})\ne0$. From
$$\Theta=2+q(1-(-1)^{\Tr_{1}^{n}((a_{1}+a_{2})(a_{1}+\overline{a_{2}}))}),$$
we have that $h_{2}(x)$ is bent if and only if $\Tr_{n}^{2n}((a_{1}+a_{2})\overline{b})=0$, $\Tr_{1}^{n}\left((a_{1}+a_{2})(a_{1}+\overline{a_{2}})\right)=0$ and $h_{2}(0)=0$. Thus we obtain conditions (2) and (3).

In the next, we prove that if $b\notin U$, then $h_{2}$ is not bent. For all $a\in \mathbb{F}_{q^2}^{*}$, we have
$$\xi(a,b)=\left(1-\mathcal{K}_{n}\left(\frac{a\overline{a}}{1+b^2\overline{b}^2}\right)\right)(-1)^{\Tr_{1}^{2n}\left(\frac{\overline{a}b}{1+b\overline{b}}\right)},$$
and then \begin{eqnarray*}
        \Theta &=& q+1+\left(1-\mathcal{K}_{n}\left(\frac{a_{1}}{1+b\overline{b}}\right)\right)(-1)^{\Tr_{1}^{2n}\left(\frac{a_{1}b}{1+b\overline{b}}\right)} \\
         &&+\left(1-\mathcal{K}_{n}\left(\frac{a_{2}\overline{a_{2}}}{1+b^2\overline{b}^2}\right)\right)(-1)^{\Tr_{1}^{2n}\left(\frac{\overline{a_{2}}b}{1+b\overline{b}}\right)}  \\
         && -\left(1-\mathcal{K}_{n}\left(\frac{(a_{1}+a_{2})(a_{1}+\overline{a_{2}})}{1+b^2\overline{b}^2}\right)\right)(-1)^{\Tr_{1}^{2n}\left(\frac{(a_{1}+\overline{a_{2}})b}{1+b\overline{b}}\right)}\\ &=&q+1+(-1)^{\Tr_{1}^{2n}\left(\frac{a_{1}b}{1+b\overline{b}}\right)}+(-1)^{\Tr_{1}^{2n}\left(\frac{\overline{a_{2}}b}{1+b\overline{b}}\right)}
         -(-1)^{\Tr_{1}^{2n}\left(\frac{(a_{1}+\overline{a_{2}})b}{1+b\overline{b}}\right)} \\
         &&-\mathcal{K}_{n}\left(\frac{a_{1}}{1+b\overline{b}}\right)(-1)^{\Tr_{1}^{2n}\left(\frac{a_{1}b}{1+b\overline{b}}\right)}
         -\mathcal{K}_{n}\left(\frac{a_{2}\overline{a_{2}}}{1+b^2\overline{b}^2}\right)(-1)^{\Tr_{1}^{2n}\left(\frac{\overline{a_{2}}b}{1+b\overline{b}}\right)}\\
         &&+\mathcal{K}_{n}\left(\frac{(a_{1}+a_{2})(a_{1}+\overline{a_{2}})}{1+b^2\overline{b}^2}\right)(-1)^{\Tr_{1}^{2n}\left(\frac{(a_{1}+\overline{a_{2}})b}{1+b\overline{b}}\right)}.
      \end{eqnarray*}
Combining the facts
$$1+(-1)^{\Tr_{1}^{2n}\left(\frac{a_{1}b}{1+b\overline{b}}\right)}+(-1)^{\Tr_{1}^{2n}\left(\frac{\overline{a_{2}}b}{1+b\overline{b}}\right)}
         -(-1)^{\Tr_{1}^{2n}\left(\frac{(a_{1}+\overline{a_{2}})b}{1+b\overline{b}}\right)}
         =2(-1)^{\Tr_{1}^{2n}\left(\frac{a_{1}b}{1+b\overline{b}}\right)\Tr_{1}^{2n}\left(\frac{\overline{a_{2}}b}{1+b\overline{b}}\right)}$$ and $|\mathcal{K}_{n}(a)|\le 2^{\frac{n}{2}+1}+1$,
we have
\begin{eqnarray*}
\Theta &>&q-2-\mathcal{K}_{n}\left(\frac{a_{1}}{1+b\overline{b}}\right)(-1)^{\Tr_{1}^{2n}\left(\frac{a_{1}b}{1+b\overline{b}}\right)}
         -\mathcal{K}_{n}\left(\frac{a_{2}\overline{a_{2}}}{1+b^2\overline{b}^2}\right)(-1)^{\Tr_{1}^{2n}\left(\frac{\overline{a_{2}}b}{1+b\overline{b}}\right)}\\
         &&+\mathcal{K}_{n}\left(\frac{(a_{1}+a_{2})(a_{1}+\overline{a_{2}})}{1+b^2\overline{b}^2}\right)(-1)^{\Tr_{1}^{2n}\left(\frac{(a_{1}+\overline{a_{2}})b}{1+b\overline{b}}\right)}\\
         &&>q-2-3(2^{\frac{n}{2}+1}+1)=2^{\frac{n}{2}}(2^{\frac{n}{2}}-6)-5>2
\end{eqnarray*}
if $n\ge 6$. Thus, $\Theta=\pm 2$ can not hold and then $h_{2}(x)$ is not bent if $n\ge 6$ for $b\in \mathbb{F}_{q^2}\setminus U$.

This finishes the proof.
\end{proof}

\begin{example}
Let $q=2^6$ and $w\in \mathbb{F}_{q^2}$ be a primitive element with minimal polynomial $m(x)=x^{12} + x^7 + x^6 + x^5 + x^3 + x + 1\in \mathbb{F}_{2}[x]$. Letting $\textsf{A}=\{(w^{195},w,1), (1,w^{258},w^{63}), (1,w^{60}, w^{63})\}$, we have that the three triples $(a_{1},a_{2},b)$ in $\textsf{A}$ satisfy conditions (1), (2) and (3), respectively, then $\Tr_{1}^{12}(\frac{a_{1}}{x^{63}+b})\Tr_{1}^{12}(\frac{a_{2}}{x^{63}+b})$ is bent for each triple $(a_{1},a_{2},b)\in \textsf{A}$.
\end{example}

 The bentness of the multiplication in Theorem \ref{thmh1} has a high interest as far as we know; it is not trivial to derive bent functions by directly multiplying two Boolean functions with the same forms. The above theorem presents such a construction.

Now we describe the third class of bent functions with $F_{3}=X_{1}X_{2}+X_{1}X_{3}+X_{2}X_{3}$.

\begin{thm}\label{thm3}
Let $q=2^n$ with $n>2$, $a_{1}\in \mathbb{F}_{q}^{*}$, $a_{2},a_{3}\in \mathbb{F}_{q^2}^{*}$ be pairwise different and $a_{1}+a_{2}+a_{3}\ne0$, and $F_{3}=X_{1}X_{2}+X_{1}X_{3}+X_{2}X_{3}$. Then $h_{3}(x)=F_{3}(f_{1},f_{2},f_{3})$ is bent if one of  the following conditions is satisfied:

\begin{itemize}
\item[-] $b=1$,

{\rm(1)} $a_{2}\in \mathbb{F}_{q}^{*}$, $a_{3}\notin \mathbb{F}_{q}$, $\Tr_{1}^{n}(a_{1}+a_{2})=1$;

%(2) $a_{3}\in \mathbb{F}_{q}^{*}$, $a_{2}\notin \mathbb{F}_{q}$, $\Tr_{1}^{n}(a_{1}+a_{3})=1$;

{\rm(2)} $a_{2}\notin \mathbb{F}_{q}$, $a_{3}\notin \mathbb{F}_{q}$, $a_{2}+a_{3}\in \mathbb{F}_{q}$, $\Tr_{1}^{n}(a_{2}+a_{3})=0$ and $\Tr_{1}^{2n}(a_{2})=0$;

\item[-] $b\in U\setminus\{1\}$,

{\rm(3)} $\Tr_{n}^{2n}(a_{2}\overline{b})=\Tr_{n}^{2n}(a_{3}\overline{b})=0$, $\Tr_{1}^{n}(a_{2}\overline{a_{2}}+a_{3}\overline{a_{3}})=1$ and $h_{3}(0)=0$;

{\rm(4)} $\Tr_{n}^{2n}(a_{2}\overline{b})=0$ and $\Tr_{n}^{2n}(a_{3}\overline{b})\ne 0$, $\Tr_{n}^{2n}((a_{1}+a_{2}+a_{3})\overline{b})=0$, $\Tr_{1}^{n}(a_{2}\overline{a_{2}})=\Tr_{1}^{n}((a_{1}+a_{2}+a_{3})(a_{1}+\overline{a_{2}}+\overline{a_{3}}))$ and $h_{3}(0)=0$;

%(6) $\Tr_{n}^{2n}(a_{3}\overline{b})=0$ and $\Tr_{n}^{2n}(a_{2}\overline{b})\ne 0$, $\Tr_{n}^{2n}((a_{1}+a_{2}+a_{3})\overline{b})=0$, $\Tr_{1}^{n}(a_{3}\overline{a_{3}})=\Tr_{1}^{n}((a_{1}+a_{2}+a_{3})(a_{1}+\overline{a_{2}}+\overline{a_{3}}))$ and $h_{3}(0)=0$;

{\rm(5)} $\Tr_{n}^{2n}(a_{2}\overline{b})\ne 0$, $\Tr_{n}^{2n}(a_{3}\overline{b})\ne 0$, $\Tr_{n}^{2n}((a_{1}+a_{2}+a_{3})\overline{b})\ne0$ and $h_{3}(0)=0$;

\item[-] $b\notin U$,

{\rm(6)} $2(-1)^{h_{3}(0)}=2(-1)^{u(a_{1})u(a_{2})+u(a_{1})u(a_{3})+u(a_{2})u(a_{3})}+(-1)^{u(a_{1}+a_{2}+a_{3})}k(a_{1}+a_{2}+a_{3})-\sum_{i=1}^{3}(-1)^{u(a_{i})}k(a_{i})$, where
$k(x)=\mathcal{K}_{n}\left(\frac{x\overline{x}}{1+b^2\overline{b}^2}\right)$ and $u(x)=\Tr_{1}^{2n}\left(\frac{\overline{x}b}{1+b\overline{b}}\right)$.
\end{itemize}
\end{thm}
\begin{proof}
The proof is similar to Theorem \ref{thmh1}, while it is a bit more complicated because there are more choices for parameters
$a_{1}$, $a_{2}$ and $a_{3}$. For $\alpha=(s_{1},s_{2},s_{3})\in \mathbb{F}_{2}^{3}$, we easily have
%$$W_{F_{3}}(\alpha)=1+(-1)^{s_{1}}+(-1)^{s_{2}}+(-1)^{s_{3}}-(-1)^{s_{1}+s_{2}}-(-1)^{s_{1}+s_{3}}-(-1)^{s_{2}+s_{3}}-(-1)^{s_{1}+s_{2}+s_{3}},$$
$$W_{F_{3}}(s_{1},s_{2},s_{3})=\left\{\begin{array}{cc}
                                       0,  & {\rm if}\,s_{1}+s_{2}+s_{3}=0, \\
                                       4(-1)^{s_{2}s_{3}},  &  {\rm if}\,s_{1}+s_{2}+s_{3}=1,
                                      \end{array}
\right.$$
%i.e.,
%\begin{eqnarray*}
%  W_{F_{3}}(000) &=&W_{F_{3}}(110)=W_{F_{3}}(101)=W_{F_{3}}(011)=0,  \\
%  W_{F_{3}}(100) &=&W_{F_{3}}(010)=W_{F_{3}}(001)=4,  \\
%  W_{F_{3}}(111) &=&-4,
%\end{eqnarray*}
which reduces Eq.\eqref{xxeq} to
\begin{equation}\label{eq3.16}
\xi(a_{1},b)+\xi(a_{2},b)+\xi(a_{3},b)-\xi(a_{1}+a_{2}+a_{3},b)=2(-1)^{h_{3}(0)}.
\end{equation}
 We also denote the left-hand side of Eq.\eqref{eq3.16} by $\Theta$ and proceed with the proof according to values of $b$.

Case $b=1$. If both $a_{2}$ and $a_{3}$ belong to $\mathbb{F}_{q}^{*}$, then Eq.\eqref{eq3.16} and Lemma \ref{lem5} give
 \begin{eqnarray*}
  \Theta &=&2+(-1)^{\Tr_{1}^{n}(a_{1})}q+(-1)^{\Tr_{1}^{n}(a_{2})}q+(-1)^{\Tr_{1}^{n}(a_{3})}q-(-1)^{\Tr_{1}^{n}(a_{1}+a_{2}+a_{3})}q  \\
    &=&2+2(-1)^{\Tr_{1}^{n}(a_{1})\Tr_{1}^{n}(a_{2})+\Tr_{1}^{n}(a_{1})\Tr_{1}^{n}(a_{3})+\Tr_{1}^{n}(a_{2})\Tr_{1}^{n}(a_{3})}q.
 \end{eqnarray*}
Hence, Eq.\eqref{eq3.16} does not hold and $h_{3}(x)$ is not bent. We assume that $a_{2}\in \mathbb{F}_{q}^{*}$ and $a_{3}\notin \mathbb{F}_{q}$. Recalling  that $a_{1}\in \mathbb{F}_{q}$, we have $a_{1}+a_{2}+a_{3}\notin \mathbb{F}_{q}$ and then
 \begin{eqnarray*}
   \Theta &=& 1+(-1)^{\Tr_{1}^{n}(a_{1})}q+1+(-1)^{\Tr_{1}^{n}(a_{2})}q+1-1 \\
    &=&  2+q\left((-1)^{\Tr_{1}^{n}(a_{1})}+(-1)^{\Tr_{1}^{n}(a_{2})}\right).
 \end{eqnarray*}
Note that $h_{3}(0)=0$. Therefore, $h_{3}(x)$ is bent if and only if $\Tr_{1}^{n}(a_{1}+a_{2})=1$. This implies condition (1). Next, we assume that neither $a_{2}$ nor $a_{3}$ belongs to $\mathbb{F}_{q}$. From
 \begin{eqnarray*}
   \Theta &=&1+(-1)^{\Tr_{1}^{n}(a_{1})}q+2-\xi(a_{1}+a_{2}+a_{3},1),
 \end{eqnarray*}
 if $h_{3}(x)$ is bent, we necessarily have $a_{1}+a_{2}+a_{3}\in \mathbb{F}_{q}^{*}$, and then $a_{2}+a_{3}\in \mathbb{F}_{q}^{*}$. It follows that $\Theta=2+q(-1)^{\Tr_{1}^{n}(a_{1})}(1-(-1)^{\Tr_{1}^{n}(a_{2}+a_{3})})$. Thus, we need $\Tr_{1}^{n}(a_{2}+a_{3})=0$, to assure the bentness of $h_{3}(x)$. Note that
 \begin{eqnarray*}
 h_{3}(0)&=&\Tr_{1}^{2n}(a_{1})\Tr_{1}^{2n}(a_{2}+a_{3})+\Tr_{1}^{2n}(a_{2})\Tr_{1}^{2n}(a_{3})\\
 &=&\Tr_{1}^{2n}(a_{2})\Tr_{1}^{2n}(a_{3})=\Tr_{1}^{2n}(a_{2})\Tr_{1}^{2n}(a_{2}+a_{2}+a_{3})\\
 &=&\Tr_{1}^{2n}(a_{2}),
 \end{eqnarray*}
due to $\Tr_{1}^{2n}(a_{2}+a_{3})=0$. Then we obtain condition (2), i.e., $h_{3}(x)$ is bent if and only if $a_{2}+a_{3}\in \mathbb{F}_{q}$, $\Tr_{1}^{n}(a_{2}+a_{3})=0$ and $\Tr_{1}^{2n}(a_{2})=0$.

Case $b\in U\setminus\{1\}$. From Eq.\eqref{eq3.4}, we know $\xi(a_{1},b)=1$. Concerning to $\xi(a_{2},b)$ and $\xi(a_{3},b)$, by Eq.\eqref{eq3.2}, we have the following subcases.
\begin{itemize}
\item Subcase {a}: $\Tr_{n}^{2n}(a_{2}\overline{b})=\Tr_{n}^{2n}(a_{3}\overline{b})=0$. Obviously $a_{2}\notin \mathbb{F}_{q}$, $a_{3}\notin \mathbb{F}_{q}$. Then $\Tr_{n}^{2n}((a_{1}+a_{2}+a_{3})\overline{b})=a_{1}(b+\overline{b})\ne0$ and
\begin{eqnarray*}
  \Theta &=& 1+1+(-1)^{\Tr_{1}^{n}(a_{2}\overline{a_{2}})}q+1+(-1)^{\Tr_{1}^{n}(a_{3}\overline{a_{3}})}q-1\\
  &=& 2+q\left((-1)^{\Tr_{1}^{n}(a_{2}\overline{a_{2}})}+(-1)^{\Tr_{1}^{n}(a_{3}\overline{a_{3}})}\right).
\end{eqnarray*}
It is clear that $h_{3}(x)$ is bent if and only if $\Tr_{1}^{n}(a_{2}\overline{a_{2}}+a_{3}\overline{a_{3}})=1$ and $h_{3}(0)=0$, which is corresponding to condition (3).

\item Subcase {b}: There is exactly one nonzero element in $\{\Tr_{n}^{2n}(a_{2}\overline{b}),\Tr_{n}^{2n}(a_{3}\overline{b})\}$. Without of loss of generality, we assume that $\Tr_{n}^{2n}(a_{2}\overline{b})=0$ and $\Tr_{n}^{2n}(a_{3}\overline{b})\ne 0$. If $\Tr_{n}^{2n}((a_{1}+a_{2}+a_{3})\overline{b})\ne 0$, then
$\Theta=2+(-1)^{\Tr_{1}^{n}(a_{2}\overline{a_{2}})}q$, and then $h_{3}(x)$ is not bent. If $\Tr_{n}^{2n}((a_{1}+a_{2}+a_{3})\overline{b})=0$, from
$$\Theta=2+q\left((-1)^{\Tr_{1}^{n}(a_{2}\overline{a_{2}})}-(-1)^{\Tr_{1}^{n}((a_{1}+a_{2}+a_{3})(a_{1}+\overline{a_{2}}+\overline{a_{3}}))}\right),$$
we deduce that $h_{3}(x)$ is bent if and only if $\Tr_{1}^{n}(a_{2}\overline{a_{2}})=\Tr_{1}^{n}((a_{1}+a_{2}+a_{3})(a_{1}+\overline{a_{2}}+\overline{a_{3}}))$ and $h_{3}(0)=0$. This gives condition (4).

\item Subcase {c}: $\Tr_{n}^{2n}(a_{2}\overline{b})\ne0$, $\Tr_{n}^{2n}(a_{3}\overline{b})\ne0$. Then we get condition (5), i.e., $h_{3}(x)$ is bent if and only if $\Tr_{n}^{2n}((a_{1}+a_{2}+a_{3})\overline{b})\ne0$ and $h_{3}(0)=0$.
\end{itemize}
 Case $b\notin U$. By Eq.\eqref{eq3.2},
$$\xi(a_{i},b)=\left(1-\mathcal{K}_{n}\left(\frac{a_{i}\overline{a_{i}}}{1+b^2\overline{b}^2}\right)\right)(-1)^{\Tr_{1}^{2n}\left(\frac{\overline{a_{i}}b}{1+b\overline{b}}\right)}$$
for $i=1,2,3$, then $h_{3}(x)$ is bent if and only if
\begin{eqnarray*}
2(-1)^{h_{3}(0)}&=&\sum_{i=1}^{3}(-1)^{u(a_{i})}(1-k(a_{i}))-(-1)^{u(a_{1}+a_{2}+a_{3})}(1-k(a_{1}+a_{2}+a_{3}))\\
&=& \sum_{i=1}^{3}(-1)^{u(a_{i})}-(-1)^{u(a_{1})+u(a_{2})+u(a_{3})}\\
&&+(-1)^{u(a_{1})+u(a_{2})+u(a_{3})}k(a_{1}+a_{2}+a_{3})-\sum_{i=1}^{3}(-1)^{u(a_{i})}k(a_{i})\\
&=&2(-1)^{u(a_{1})u(a_{2})+u(a_{1})u(a_{3})+u(a_{2})u(a_{3})}\\
&&+(-1)^{u(a_{1})+u(a_{2})+u(a_{3})}k(a_{1}+a_{2}+a_{3})-\sum_{i=1}^{3}(-1)^{u(a_{i})}k(a_{i}).
\end{eqnarray*}
The proof is therefore completed.
\end{proof}
\begin{example}
Let $q=2^3$ and $w\in \mathbb{F}_{q^2}$ be the same as in Example 1. Choose $\textsf{A}=\{(w^{27},w^9,w,1),\\ (w^{27},w,w^5,1), (w^{27}, w^{16}, w^{34}, w^7),(w^{27}, w^{16}, w^{11}, w^7),(w^{27}, w, w^{3}, w^{7}),(w^{27}, w^{3}, w^9, w)\}$, then the quadruples $(a_{1},a_{2},a_{3},b)$ in $\textsf{A}$ satisfy the conditions (1)-(6), respectively. It can be verified that $h_{3}$ is bent for each quadruple in $\textsf{A}$.
\end{example}

\begin{remark}
We divided the conditions into three cases in Theorem \ref{thm3}. The characterizations on $a_{1},a_{2},a_{3}$ and $b$ such that $h_{3}(x)$ is bent in cases $b=1$ and $b\in U\setminus\{1\}$ are clear. In the case $b\notin U$, we give a relation on these four Kloosterman sums, which is far from clear and not satisfying. If we know more about these four sums, we might give a more clear characterization of the coefficients $a_{1},a_{2},a_{3}$ and $b$. We omit the description of conditions which can be easily obtained by the symmetry of $a_{2}$ and $a_{3}$, so all the possibilities of these parameters have been considered. In this sense, the conditions in Theorem \ref{thm3} are sufficient and necessary.
\end{remark}
\begin{remark}
Based on the functions in set $\mathcal{H}$, it is possible to choose $F$ with more variables or higher algebraic degrees, and we can follow the same ideas to obtain characterizations on other classes of bent functions. For example, if let $F=X_{1}X_{2}X_{3}$ and $h=f_{1}f_{2}f_{3}$, then
$$\left\{\begin{array}{l}
   W_{F}(000)=6,  \\
   W_{F}(100)=W_{F}(010)=W_{F}(001)=W_{F}(111)=2,  \\
   W_{F}(110)=W_{F}(101)=W_{F}(011)=-2.
\end{array}\right.$$
By Eq.\eqref{xxeq}, we obtain that $h$ is bent if and only if
\begin{eqnarray*}
&&4(-1)^{\Tr_{1}^{m}\left(\frac{a_{1}}{b}\right)\Tr_{1}^{m}\left(\frac{a_{2}}{b}\right)\Tr_{1}^{m}\left(\frac{a_{3}}{b}\right)}\\
&=&3(q+1)+\xi(a_{1},b)+\xi(a_{2},b)+\xi(a_{3},b)+\xi(a_{1}+a_{2}+a_{3},b)\\
&&-\xi(a_{1}+a_{2},b)-\xi(a_{1}+a_{3},b)-\xi(a_{2}+a_{3},b).
\end{eqnarray*}
Hence, more thorough investigations are needed to determine the parameters' characterisation.
\end{remark}

\section{On EA Equivalence}\label{sec4}
Two Boolean functions $\mu(x)$ and $\nu(x)$ on $\mathbb{F}_{2^m}$ are called Extended-Affine equivalent  (in brief, EA-equivalent), if there exists an affine permutation $\mathcal{A}$ and an affine Boolean function $\mathcal{L}$, such that $\mu=\nu\circ \mathcal{A}+\mathcal{L}$. If a class of bent functions is proposed, we want to know whether it is inequivalent with existing classes of functions. However, there is in general, no effective ways to decide whether two bent functions are
EA-equivalent. The EA equivalences between many known classes of concrete constructions of bent functions are unclear.

 Since the algebraic degrees of all Dillon-like bent functions are $\frac{m}{2}$, we know they are not EA-equivalent to these Gold, Leander, Canteaut-Charpin-Kyureghyan monomial bent functions.
 The function set $\mathcal{H}$ constitute the basis of bent functions in our constructions, the class of $h_{1}$ comes directly from $\mathcal{H}$, while the classes of $h_{2}$ and $h_{3}$ are obtained from $\mathcal{H}$ in nonlinear ways. By the fact that Kasami bent functions on $\mathbb{F}_{2^{2n}}$ with algebraic degrees $n$ exist only for even $n$, we know that the class of $h_{1}$ is not EA-equivalent with the class of Kasami functions. The experiments on $\mathbb{F}_{2^8}$ shows that the class of $h_{1}$ is not EA-equivalent to Dillon-like monomial functions. Therefore, the class of functions $h_{1}$ are not EA-equivalent to all these known monomial bent functions. For the two types of Boolean bent functions
$$\sum_{i\in D}\Tr_{1}^{m}(ax^{(ri+s)(q-1)})+\Tr_{1}^{o(d)}(\epsilon x^{\frac{q^2-1}{d}})$$
and
$$\sum_{i=0}^{d-1}\Tr_{1}^{m}(a_{i}x^{(l+i\frac{q+1}{d})(q-1)})+\Tr_{1}^{o(d)}(\epsilon x^{\frac{q^2-1}{d}})$$
 in \cite{nianli,yulong}, the class of functions $h_{1}$ is not EA-equivalent with them by experiments on $\mathbb{F}_{2^6}$. In \cite[Theorem 2]{nianli}, the authors discuss the bentness of a general trace quadrinomials as the following
 \begin{equation}\label{eq5.1}
 \Tr_{1}^{m}(ax^{l(q-1)}+bx^{(\frac{q+1}{e}-1)(q-1)}+cx^{(\frac{q+1}{e}+1)(q-1)})+\Tr_{1}^{t}(\epsilon x^{\frac{q^2-1}{e}}).
 \end{equation}
  We note that all the Dillon exponents are divided in to two, two and four cyclotomic cosets on $\mathbb{F}_{2^6}$, $\mathbb{F}_{2^8}$ and $\mathbb{F}_{2^{10}}$, respectively, then all Dillon-like bent functions are covered by Eq.\eqref{eq5.1} on these fields. For larger field $\mathbb{F}_{2^m}$ with $m\ge 12$, checking the EA inequivalence becomes infeasible on personal computer. Therefore, the equivalence between the class of functions $h_{1}$ and functions in Eq.\eqref{eq5.1} remains unclear. However, if $h_{1}$ is bent, its ordinary polynomial form can be written explicitly (see details in Appendix A1) as
\begin{equation}\label{eq4-2}
\left\{
\begin{array}{ll}
\Tr_{1}^{m}\left(ab^{-1}\sum_{j=1,j\,odd}^{q}b^{-j}x^{j(q-1)}\right)+\Tr_{1}^{m}(ab^{q^2-2}),&{\rm if}\, b\in U,\\
\Tr_{1}^{m}\left(\frac{a\overline{b}}{1+b\overline{b}}\sum_{j=1}^{q}b^{-j}x^{j(q-1)}\right)+\Tr_{1}^{m}(ab^{q^2-2}),&{\rm if}\, b\notin U.
\end{array}\right.
\end{equation}
We see that the ordinary polynomial form of $h_{1}$ is full with all Dillon exponents and is different with these known Dillon-like bent functions having fixed number of trace terms.

\section{Conclusion}\label{sec5}
It is an important direction that new functions can be obtained in a cryptographic sense based on suitable families of block functions. This paper introduces a special family of Dillon-like Boolean functions $\mathcal{H}$. These functions are trace rational and the Walsh transform can be determined. By using functions in $\mathcal{H}$ as building blocks, we obtain three classes of bent functions, of which the characterizations also are connected with the well-known Kloosterman sum. We deduce that bent functions not EA-equivalent to all known classes of monomial bent functions are included in our function classes. The explicit ordinary polynomial form in Eq.\eqref{eq4-2} shows that the class of $h_{1}$ is different with these classes of Dillon-like bent functions with fixed number of trace terms. Following the ideas in this paper, new families of block functions are worth considering, and new bent functions may be generated.

\section*{Appendix A1:  The Explicit Ordinary Polynomial Form of $h_{1}$}
To deduce the explicit polynomial form of $h_{1}$, a basic fact is that for integers $i$ and $d$ with $i\ge0$ and $d>0$, if $i=ed+r$, then $$x^i\equiv x^{e+r} \,{\rm mod}\, (x^d-x).$$
 Write every integer $0\le i\le q^2-2$ as $i=k(q+1)+j$, with $j\le q$ and $0\le k\le q-2$. If $j=0$, then
$$x^{i(q-1)}=x^{k(q+1)(q-1)}\equiv \left\{\begin{array}{ll}
     1 \,{\rm mod}\, (x^{q^2}-x), & {\rm if} \,k=0, \\
    x^{q^2-1} \,{\rm mod}\, (x^{q^2}-x),& {\rm if}\, k\ge 1.
  \end{array}\right.
$$
If $j\ne 0$, then
$$x^{i(q-1)}=x^{k(q^2-1)+j(q-1)}=x^{kq^2+j(q-1)-k}\equiv x^{j(q-1)} \,{\rm mod}\,(x^{q^2}-x).$$
Since ${q^2-2 \choose i}\equiv 1 \,{\rm mod}\,2$ if and only if $i$ is even, we get
\begin{eqnarray*}
\frac{1}{x^{q-1}+b}&=&\sum_{i=0}^{q^2-2}{q^2-2 \choose i}b^{q^2-2-i}x^{i(q-1)}
= \sum_{i=0,i\, {even}}^{q^2-2}b^{q^2-2-i}x^{i(q-1)} \\
&=&\sum_{j=0}^{q}\sum_{k=0, k(q+1)+j\, even}^{q-2}b^{q^2-2-k(q+1)-j}x^{k(q^2-1)+j(q-1)}\\
&\equiv&b^{q^2-2}+\sum_{k=1,k \,even}^{q-2}b^{q^2-2-k(q+1)}x^{q^2-1}\\
&&+\sum_{j=1}^{q}\left(\sum_{k=0,k(q+1)+j\, even}^{q-2}b^{q^2-2-k(q+1)-j}\right)x^{j(q-1)} \,{\rm mod}\, (x^{q^2}-x).
\end{eqnarray*}
Note that $\sum_{k=0}^{q-2}{b^{-k(q+1)}}=0$ for any $b\notin U$, then
\begin{eqnarray*}
\sum_{k=0,k\,even}^{q-2}b^{-k(q+1)}
%&=&\sum_{k=0,k\, odd}^{q-2}b^{-k(q+1)}\\
&=&1+b^{-2(q+1)}+b^{-4(q+1)}+\ldots+b^{-(q-2)(q+1)}\\
&=&\frac{1+b^{-q(q+1)}}{1+b^{-2(q+1)}}\\
&=&\frac{1+b^{-(q+1)}}{1+b^{-2(q+1)}}\\
&=&\frac{b\overline{b}}{1+b\overline{b}}
\end{eqnarray*}
and
$$\sum_{k=1,k\,even}^{q-2}b^{-k(q+1)}=\frac{1}{1+b\overline{b}}.$$
Hence, the coefficient of term $x^{q^2-1}$ equals
$$\sum_{k=1,k \,even}^{q-2}b^{q^2-2-k(q+1)}=\left\{\begin{array}{cc}
                                              \frac{1}{b}, & {\rm if}\, b\in U, \\
                                               \frac{1}{b(1+b\overline{b})},&  {\rm if}\,b\notin U.
                                            \end{array}\right.
$$
When $h_{1}$ is bent, since the algebraic degree of a bent function is at most $\frac{m}{2}$, then $\Tr_{1}^{m}(\frac{a}{b})=0$ if $b\in U$ and $\Tr_{1}^{m}\left(\frac{a}{b(1+b\overline{b})}\right)=\Tr_{1}^{n}\left(\frac{a(b+\overline{b})}{b\overline{b}(1+b\overline{b})}\right)=0$ if $b\notin U$, which coincide with partial requirements in Theorem \ref{thm1}. Further, since
$$b^{-1-j}\sum_{k=0,k(q+1)+j \, even}^{q-2}b^{-k(q+1)}=\left\{\begin{array}{ll}
b^{-j-1},& {\rm if}\,j\, {\rm odd},\\
0,& {\rm if}\,j\, {\rm even}
\end{array}\right.
$$
for $b\in U$, and
 $$b^{-1-j}\sum_{k=0,k(q+1)+j \, even}^{q-2}b^{-k(q+1)}=
b^{-j-1}\frac{b\overline{b}}{1+b\overline{b}}
$$
for $b\notin U$, then Eq.\eqref{eq4-2} follows.

\end{document}